\newif\ifsa 
\satrue

\documentclass{article} 
\usepackage[utf8]{inputenc}
\usepackage[english]{babel}
\usepackage{amsmath,amsfonts}
\usepackage{amssymb}
\usepackage{amsthm}
\usepackage[letterpaper, portrait, left=1in, right=1in, top=1in, bottom=1in]{geometry}
\usepackage[hyperindex,breaklinks,bookmarks]{hyperref}
\hypersetup{bookmarksdepth=2}
\usepackage{xspace}
\usepackage{cleveref}
\usepackage{soul}
\usepackage{nicefrac}
\usepackage{thm-restate}
\usepackage[shortlabels]{enumitem}
\usepackage[font=small,labelfont=bf]{caption}
\usepackage[normalem]{ulem}
\usepackage{wrapfig}
\usepackage{enumitem}

\usepackage{xcolor}
\usepackage{hyperref}
\hypersetup{
    colorlinks,
    linkcolor={red!50!black},
    citecolor={blue!50!black},
    urlcolor={blue!80!black}
}

\usepackage[ 
backend=bibtex,maxbibnames=99,sortcites]{biblatex}
\newtheorem{theorem}{Theorem}
\newtheorem*{theorem*}{Theorem}
\newtheorem{corollary}[theorem]{Corollary}
\newtheorem*{corollary*}{Corollary}

\newtheorem*{claim}{Claim}
\newtheorem{lemma}[theorem]{Lemma}

\crefname{theorem}{Theorem}{Theorems}
\crefname{lemma}{Lemma}{Lemmas}
\crefname{proposition}{Proposition}{Propositions}
\crefname{claim}{Claim}{Claims}
\crefname{corollary}{Corollary}{Corollaries}
\crefname{definition}{Definition}{Definitions}
\crefname{figure}{Figure}{Figures}
\crefname{obs}{Observation}{Observations}
\crefname{equation}{Equation}{Equations}
\crefname{section}{Section}{Sections}

\theoremstyle{remark}

\newcommand{\ld}{\left}
\newcommand{\rd}{\right}
\newcommand{\onlinesorter}[1]{\texttt{OnlineSorter}_{#1}}
\newcommand{\onlinepacker}{\texttt{OnlinePacker}}

\newcommand{\A}{\mathcal{A}}
\newcommand{\R}{\mathcal{R}}

\newcommand{\D}{\mathcal{D}}

\newcommand{\E}{\mathcal{E}}

\newcommand{\cost}{\textrm{cost}}

\newcommand{\OnlineSorting}{\textsc{Online-Sorting}\xspace}
\newcommand{\StripPackingConv}{\textsc{Strip-Packing}\xspace}
\newcommand{\SquarePacking}{\textsc{Square-Packing}\xspace}
\newcommand{\PerimeterPacking}{\textsc{Perimeter-Packing}\xspace}
\newcommand{\BinPacking}{\textsc{Bin-Packing}\xspace}

\newcommand{\ALG}{\text{ALG}}
\newcommand{\OPT}{\text{OPT}}
\newcommand{\Area}[1]{\text{area}(#1)}

\DeclareMathOperator{\pwidth}{width}
\DeclareMathOperator{\pextwidth}{extwidth}
\DeclareMathOperator{\pheight}{height}
\DeclareMathOperator{\area}{area}
\newcommand{\mydef}{:=}

\newcommand\blfootnote[1]{%
  \begingroup
  \renewcommand\thefootnote{}\footnote{#1}%
  \addtocounter{footnote}{-1}%
  \endgroup
}

\date{April 2024}
\title{Online Sorting and Translational Packing of Convex Polygons\footnote{This paper was presented at the ACM-SIAM Symposium on Discrete Algorithms (SODA23).}}

\author{Anonymous}

\ifsa
\usepackage{authblk}

\author[1]{Anders Aamand$^\dagger$}
\author[1]{Mikkel Abrahamsen$^\dagger$}
\author[1]{Lorenzo Beretta$^\dagger$}
\author[2]{Linda Kleist$^\dagger$}
\affil[1]{BARC, University of Copenhagen, \texttt{\{aa,miab,beretta\}@di.ku.dk}}
\affil[2]{Technische Universtität Braunschweig, \texttt{kleist@ibr.cs.tu-bs.de}}
\fi

\usepackage{graphicx}
\graphicspath{{./figures/}}

\begin{document}
\maketitle

\ifsa
\blfootnote{
\begin{wrapfigure}{l}{0.055\textwidth}
\vspace{-6 mm}
\includegraphics[height=0.055\textwidth, width=0.055\textwidth]{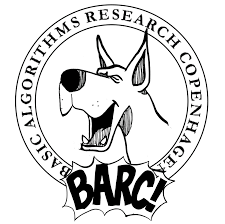}\\
\vspace{-1mm}
\includegraphics[height=0.03\textwidth, width=0.055\textwidth]{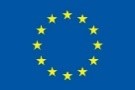}
\end{wrapfigure}
$^\dagger$ Anders Aamand is funded by the DFF-International Postdoc Grant 0164-00022B from the Independent Research Fund Denmark.
Anders Aamand, Mikkel Abrahamsen, and Lorenzo Beretta are part of Basic Algorithms Research Copenhagen (BARC), supported by the VILLUM Foundation grant 16582.
Mikkel Abrahamsen is supported by Starting Grant 1054-00032B from the Independent Research Fund Denmark under the Sapere Aude research career programme.
Lorenzo Beretta receives funding from the European Union's Horizon 2020 research and innovation program under the Marie Skłodowska-Curie grant agreement No.~801199.
Linda Kleist was able to visit BARC by the support of a  postdoc fellowship of the German Academic Exchange Service (DAAD).}
\fi

\begin{abstract}
We investigate several online packing problems in which convex polygons arrive one by one and have to be placed irrevocably into a container, while the aim is to minimize the used space.
Among other variants, we consider strip packing and bin packing, where the container is the infinite horizontal strip $[0,\infty)\times [0,1]$ or a collection of $1 \times 1$ bins, respectively.

If polygons may be rotated, there exist $O(1)$-competitive \emph{online} algorithms for all problems at hand [Baker and Schwarz, SIAM J. Comput., 1983].
Likewise, if the polygons may not be rotated but only translated, 
then using a result from [Alt, de Berg and Knauer, JoCG, 2017] we can derive $O(1)$-approximation algorithms for all problems at hand.
Thus, it is natural to conjecture  that the online version of these problems,  in which only translations are allowed, also admits a $O(1)$-competitive algorithm.
We disprove this conjecture by showing a superconstant lower bound on the competitive ratio for several online packing problems. 

The offline approximation algorithm for translation-only packing sorts the convex polygons by their ``natural slope,'' so that they form a fan-like pattern.
We prove that this step is essential, in the sense that packing polygons without rotating them is as hard as sorting numbers online. 
Technically, we prove lower bounds on the competitive ratio of translation-only online packing problems by reducing from 
a purpose-built  novel and natural combinatorial problem that we call \emph{online sorting}.
In a nutshell, the problem requires us to place $n$ numbers $x_1, \dots, x_n$ coming online into an oversized array of length $\gamma n$ for $\gamma\geq 1$, while minimizing the sum of absolute differences of consecutive numbers.
Note that the offline optimum is achieved by sorting $x_1, \dots, x_n$.
We show a superconstant lower bound on the competitive ratio of online sorting, for any constant $\gamma$.
We prove that this yields superconstant lower bounds for all packing problems at hand.
We believe that this technique is of independent interest since it uncovers a deep connection between inherently geometrical and purely combinatorial problems.

As a complement, we also include algorithms for both online sorting and translation-only online strip packing with non-trivial competitive ratios.
Our algorithm for strip packing relies on a new technique for recursively subdividing the strip into parallelograms of varying height, thickness and slope.
\end{abstract}

\setcounter{page}{0}
\thispagestyle{empty}

\newpage

\section{Introduction}
Packing problems are omnipresent in our daily lives and likewise appear in many large-scale industries.
For instance, two-dimensional versions of packing arise when a given set of pieces has to be cut out from a large piece of material such that waste is minimized. 
This is relevant in clothing production where pieces are cut out from a strip of fabric, and similarly in leather, glass, wood, and sheet metal cutting. 
In the 1965 American classic \emph{The Sound of Music}, Maria sewed playclothes for the seven von Trapp children, Liesl, Friedrich, Louisa, Kurt, Brigitta, Marta, and Gretl, using an old set of curtains. 
Undoubtedly, this task did require careful planning of how to cut out the individual pieces of curtain fabric in order not to run out.
In practical settings, it is often important that the inherent structure of the host material (grain of fabric, patterns, etc.)~is respected, i.e., the pieces should not be arbitrarily rotated, but merely translated.
In sewing patterns, the orientation is indicated by so-called \emph{grain line arrows}, like in \Cref{fig:Sound-of-Music} (right).

\begin{figure}[htb]
\centering
\includegraphics[scale=.25]{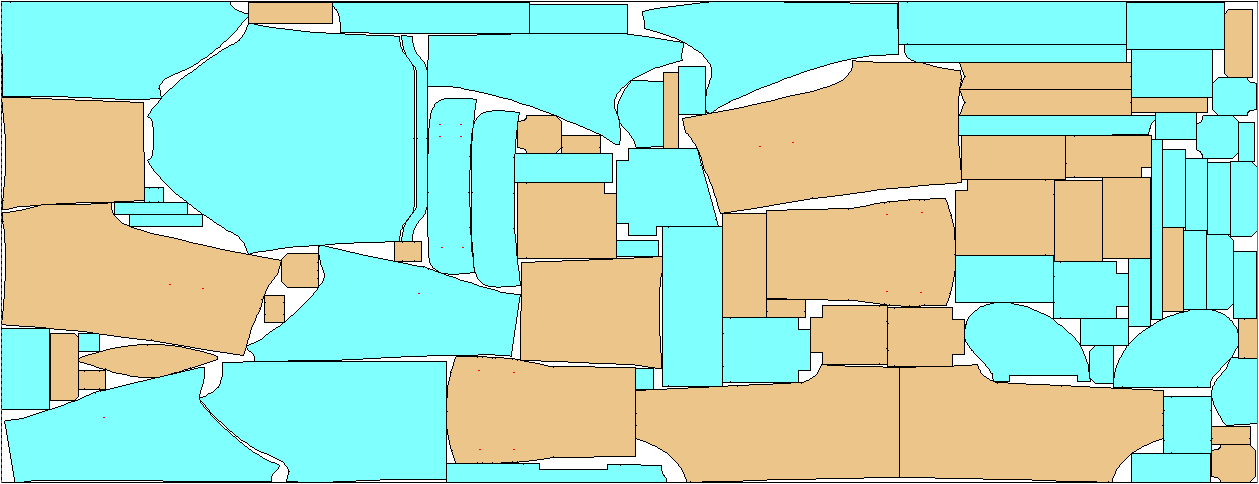}\hfill
\includegraphics[page=3,scale=.65]{figures/Sound-of-Music2.pdf}
\caption{
Left: Real-world example of packing on fabric produced using industrial software from Mirisys.
Right: Potential sewing patterns for the playclothes arranged on top of a curtain.}
\label{fig:Sound-of-Music}
\end{figure}

The objects may in some applications appear in an \emph{online} fashion, i.e., the pieces are given one after the other, and each of them must be placed before the next one is known. For example, in a scene which was not included in the final cinema edition of~\emph{The Sound of Music}, the butler of the von Trapp family was singing the measurements of the children to Maria from the kitchen (in German), and all the while she was cutting out the pieces for the playclothes one by one in the living room.
This is in contrast to \emph{offline} problems, where all the pieces are known in advance.
Problems related to packing were some of the first for which online algorithms were described and analyzed.
Indeed, the first use of the terms ``online'' and ``offline'' in the context of approximation algorithms was in the early 1970s and used for algorithms for bin-packing problems~\cite{Fiat1998}.

Most existing theoretical research on packing, and all research on online translational packing that we are aware of, is concerned with axis-parallel rectangular pieces.
In this paper, we study online translational packing of convex polygons.  The pieces arrive one by one and have to be placed irrevocably into a horizontal strip (or into bins, a square, the plane) before the next piece is revealed, and only translations of the pieces are allowed. The aim is to minimize the used space depending on the specific problem at hand, e.g., the used length of the strip, the number of bins, etc.

We note that these problems have known $O(1)$-competitive algorithms if the arriving pieces are axis-parallel rectangles; see \cref{sec:pack} for references. In other words, online algorithms exists which produce solutions that are only a constant factor worse than the offline optimum. If the arriving pieces are convex polygons that may be arbitrarily rotated, the task reduces to
packing axis-parallel rectangles by first rotating each piece so that a diameter of the piece is horizontal.
Then the density of the piece in its axis-parallel bounding box is at least $1/2$, and the algorithm for rectangles can be applied to the bounding box, again leading to $O(1)$-competitive algorithms.
We also note that using known results, it is easy to show that the \emph{offline} versions of the translational versions of these problems have $O(1)$-approximation algorithms; see \cref{thm:offline}. It is therefore natural to conjecture that also online translational packing of convex polygons admits $O(1)$-competitive algorithms. The main contribution of this paper is to show that surprisingly this conjecture is false.
In other words, whatever clever strategy Maria would use for placing of the next piece, there exist convex pieces for playclothes that could well fit 
on a single curtain
such that if she were to cut them out in an online fashion as directed by the butler, she would come short even if she had access to all the curtains in the entire von Trapp villa or even in the world.
{Remarkably, this even holds if all the pieces have diameter at most $1$\,cm or any other arbitrarily small fixed size.}

To develop lower bounds for these packing problems, we introduce the problem $\OnlineSorting[\gamma,n]$ which we believe to be of independent interest.
In this problem, we have an empty array $A$ with $\lfloor \gamma n\rfloor$ cells, $\gamma \geq 1$, and receive a stream of real numbers $s_1,\ldots,s_n$, $s_i\in[0,1]$. Each real has to be placed into  an empty cell of $A$ before the next real is known.
The goal is to minimize the sum of differences of consecutive reals in $A$.
The offline optimum is  obtained by placing the reals in sorted order in some $n$ cells of $A$.
We show that \OnlineSorting does not allow for constant factor competitive online algorithms. 

\begin{restatable}{theorem}{TheLowerBound} \label{thm:TheLowerBound}
Suppose that $\gamma,\Delta\geq 1$ are such that \OnlineSorting{}$[\gamma,n]$ admits a $\Delta$-competitive algorithm, then $\gamma\Delta = \Omega(\log n / \log \log n)$.
\end{restatable}

We then use this insight to show that various packing problems do not allow for constant factor asymptotically competitive online algorithms.
In \StripPackingConv, we have a horizontal strip of height $1$ which is bounded to the left by a vertical segment and unbounded to the right.
The goal is to place the pieces so that we use a part of the strip of minimum length.
In \BinPacking, the pieces have to be placed in unit squares, and the goal is to use a minimum number of these squares.
In \PerimeterPacking, we can place the pieces anywhere in the plane, and the goal is to minimize the perimeter of their axis-parallel bounding box.
In $\SquarePacking[\delta]$, we receive a stream of pieces with diameter at most $\delta$ and total area at most $\delta$, and the goal is to place them in a unit square.
For more background on each of these packing problems and their relation to previous work, we refer to \cref{sec:pack}.

\begin{restatable}{theorem}{PackingMasterThm}\label{thm:PackingMasterThm}
The following holds, where $n$ is the number of pieces:
\begin{enumerate}[(a)]
\item \label{MP:a} \StripPackingConv\ does not allow for a competitive online algorithm, even if all pieces have diameter at most $\delta$ for any constant $\delta>0$.
In particular, the competitive ratio of any algorithm is $\Omega(\sqrt{\log n/\log\log n})$.

\item \label{MP:b} \BinPacking\ does not allow for a competitive online algorithm, even if all pieces have diameter at most $\delta$ for any constant $\delta>0$.
In particular, the competitive ratio of any algorithm is $\Omega(\sqrt{\log n/\log\log n})$.

\item \label{MP:d} \PerimeterPacking\ does not allow for a competitive online algorithm, even if all pieces have diameter at most $\delta$ for any constant $\delta>0$.
In particular, the competitive ratio of any algorithm is $\Omega(\sqrt[4]{\log n/\log\log n})$.

\item \label{MP:c} $\SquarePacking[\delta]$ does not allow for an online algorithm for any $\delta \in (0,1]$.
In particular, for any algorithm and infinitely many $n$, there exists a stream of $n$ pieces of total area $O(\sqrt{\log\log n/\log n})$ that the algorithm cannot pack in the unit square.
\end{enumerate}
Here, (a) and (b) even hold in the asymptotic sense, i.e., if we restrict ourselves to instances with offline optimal cost at least $C$, for any constant $C>0$.
\end{restatable}
As indicated above, in the offline setting, all the problems in~\cref{thm:PackingMasterThm} have $O(1)$-approximations (\cref{thm:offline}) and also allow for $O(1)$-competitive algorithms in the online setting if the arriving pieces are axis-parallel rectangles or if rotations are allowed (see \cref{sec:pack}).

On the positive side, we present online algorithms for both online sorting and strip packing.
For \OnlineSorting{}$[\gamma,n]$, we distinguish two scenarios: the case without any extra space, i.e., $\gamma=1$, and the case $\gamma=1+\varepsilon$ a constant $\varepsilon>0$. In the case $\gamma=1$, we can provide an asymptotically tight analysis.

\begin{restatable}{theorem}{SortingUpperBoundBaby}\label{thm:SortingUpperBoundBaby}
There exists an online algorithm for \OnlineSorting{}$[1,n]$ with competitive ratio at most $18\sqrt{n}$. Every online algorithm of \OnlineSorting{}$[1,n]$ has competitive ratio at least $\sqrt{n/2}$.
\end{restatable}

As we describe in \cref{sec:tsp}, this can be seen as an asymptotically tight analysis of an online version of the travelling salesperson problem (TSP) on the real line. Indeed, we can imagine that we must visit $n$ cities on $[0,1]$ at time steps $1,\dots,n$. The position of each city is revealed to us in an online fashion and we immediately have to decide the time step where we visit this city.
In addition to packing and TSP, we believe that the online sorting problem can be useful when studying other online problems as well.

In contrast to \cref{thm:SortingUpperBoundBaby}, when the available space is a constant factor larger than $n$, there exists an algorithm with competitive ratio $n^{o(1)}$.
\begin{restatable}{theorem}{SortingUpperBound}\label{thm:SortingUpperBound}
Given a sufficiently large $n$ and $ \varepsilon \in [\log n / n , 1]$, there exists an algorithm for \OnlineSorting{}$[1+\varepsilon,n]$  with competitive ratio $2^{O\ld(\sqrt{\log n} \cdot \sqrt{\log \log n + \log (1/\varepsilon})\rd)}$.
\end{restatable}
We note that there is an exponential gap between the lower and upper bounds in~\cref{thm:TheLowerBound} and~\cref{thm:SortingUpperBound}. It is an interesting open problem to close this gap, say for  \OnlineSorting{}$[2,n]$.

There is a trivial $n$-competitive algorithm \StripPackingConv that places each of the $n$ pieces as deep into the strip as possible.
Improving upon this turns out to be quite challenging.
We present an online algorithm with competitive ratio $O(n^{\log 3-1}\log n)=O(n^{0.59})$, where $\log x$ denotes the base-$2$ logarithm of $x$.
The algorithm relies on a new technique for recursively subdividing the strip into parallelograms of varying height, thickness and slope. Each piece is then placed in a parallelogram of a suitable form and size.

\begin{restatable}{theorem}{PackingUpperBound}\label{thm:PackingUpperBound}
There exists an algorithm for \StripPackingConv\ with competitive ratio $O(n^{\log 3-1}\log n)$, where $n$ is the number of pieces.
\end{restatable}
Another interesting open problem is to improve upon this.
Is it, for example, possible to obtain an $n^{o(1)}$-competitive algorithm for \StripPackingConv as we have for \OnlineSorting{}$[1+\varepsilon,n]$?
Because the sorting problem is much simpler than the packing problem, the lower bound from \cref{thm:TheLowerBound} implies a lower bound for \StripPackingConv, but the algorithm behind~\cref{thm:SortingUpperBound} does not lead to any packing algorithm.

\subsection{The necessity of sorting pieces by slope}
Our results in \cref{thm:PackingMasterThm} are in contrast to translational offline packing of convex polygons for which constant factor approximations exist. 
In a recent paper, Alt, de Berg, and Knauer~\cite{alt_convexOffline_JoCG,alt_convexOffline_corr} gave a constant-factor approximation algorithm for offline translational packing of convex polygons so as to minimize the area of their bounding box.
The algorithm works by first grouping the pieces into exponentially increasing height classes and then sorting the pieces in each height class by the slopes\footnote{The slopes are computed with respect to the $y$-axis, namely they are the inverse of regular slopes.} of their \emph{spine segments}; see \cref{fig:fanSorting}. The spine segment of a piece is the line segment from the bottommost to the topmost corner.
Placing the pieces in rows in this sorted order (so that each row appears as a fan-like pattern) yields a compact packing with constant density.

\begin{figure}[htb]
\centering
\includegraphics[page=10]{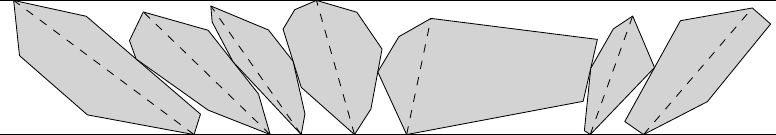}
\caption{A fan-like packing of pieces of nearly equal height.
The pieces are sorted according to the slope of their spine segments (dashed).}
\label{fig:fanSorting}
\end{figure}

We show that similar procedures yield constant factor approximations for the offline version of strip packing, bin packing, square packing and perimeter packing.

\begin{restatable}{theorem}{offline}\label{thm:offline}
There are polynomial-time offline algorithms for the following packing problems:
\begin{enumerate}[(a)]
\item \label{OP:a}
\textsc{Offline}-\StripPackingConv, $51$-approximation algorithm.
\item \label{OP:b} \textsc{Offline}-\BinPacking, $11$-approximation algorithm if the diameter of all pieces is bounded by $1/10$.\footnote{We note that we do not know if there is also a $O(1)$-factor approximation algorithm without a bound on the diameter.}
\item \label{OP:d}
\textsc{Offline}-\PerimeterPacking, $7.3$-approximation algorithm.
\item \label{OP:c}
\textsc{Offline}-\SquarePacking{}$[1/10]$, in particular, every set of convex polygons of diameter and total area at most $1/10$ can be packed into the unit square.
\end{enumerate}
\end{restatable}

The contrast between \cref{thm:PackingMasterThm} and \cref{thm:offline} indicates that sorting the pieces by the slope of their spine segments is essential for obtaining an efficient packing.
In particular, we use the lower bound from \cref{thm:TheLowerBound} to create an adaptive stream of pieces that will force any packing algorithm to use excessive space.
In the reduction, the numbers to be sorted in the online sorting problem correspond to the slopes of the spine segments in the packing problems, and the impossibility of placing the numbers in nearly sorted order implies that the packing algorithm is forced to produce an arrangement that is far from an optimal fan-like pattern.

\subsection{Relation to previous work on packing}\label{sec:pack}

The literature on online packing problems is vast.
See the surveys of Christensen, Khan, Pokutta, and Tetali~\cite{CHRISTENSEN201763}, Epstein and van Stee~\cite{epstein2018multidimensional}, van Stee~\cite{DBLP:journals/sigact/Stee12,DBLP:journals/sigact/Stee15}, and Csirik and Woeginger~\cite{Csirik1998} for an overview.
Below we survey the most important results related to the packing problems studied in this paper.
Let us highlight that when the pieces are restricted to axis-parallel rectangles, there are online algorithms with constant competitive ratios solving all the problems of \cref{thm:PackingMasterThm}.

\paragraph{Strip packing.}
In strip packing, we have a horizontal strip of height $1$ bounded by a vertical segment to the left but unbounded to the right.
The goal is to place the pieces in the strip so as to minimize the length of the part of the strip that has been used.
Milenkovich~\cite{DBLP:conf/stoc/Milenkovic96,DBLP:journals/algorithmica/Milenkovic97} and Daniels and Milenkovich~\cite{milenkovich1999translational,DBLP:journals/algorithmica/DanielsM97} described exact and approximation algorithms for offline strip packing where the pieces are simple or convex polygons.

Baker and Schwarz~\cite{baker1983shelf} described the first algorithms for online strip packing of rectangular pieces.
The FFS (First Fit Shelf) algorithm has a competitive ratio of $6.99$ under the assumption that the width of each rectangle is at most $1$.
Ye, Han, and Zang~\cite{YeOnlineStrip}
 improved the algorithm and obtained a competitive ratio of $\nicefrac{7}{2}+\sqrt{10}\approx 6.6623$ without the restriction on rectangle heights. 
Restricting the attention to large instances, FFS has an asymptotic competitive ratio that can be made arbitrarily close to $1.7$.
Csirik and Woeginger~\cite{Csirik1997shelf} described an improved algorithm with an asymptotic competitive ratio arbitrarily close to $h_\infty\approx 1.69103$.
This was later improved to $1.58889$ by Han, Iwama, Ye, and Zhang~\cite{han2007strip}.
In contrast, we show that when the pieces are convex polygons, then no competitive algorithm exists (\cref{thm:PackingMasterThm}~\ref{MP:a}).

\paragraph{Bin packing.}
In bin packing, we have an unbounded supply of identical containers, and the goal is to pack the pieces into as few containers as possible.
As mentioned, online bin packing problems have been studied since the early 1970s~\cite{Fiat1998}.
Many papers have been devoted to the study of online translational bin packing axis-parallel rectangular pieces into unit square bins.
In long sequences of papers, the upper bound on the asymptotic competitive ratio for this problem has been decreased from $3.25$ to $2.5545$ and the lower bound has been increased from $1.6$ to $1.907$~\cite{han2011new}.
In this paper, we show that when packing convex polygons instead of axis-parallel rectangles, there is no competitive algorithm (\cref{thm:PackingMasterThm}~\ref{MP:b}).

\paragraph{Perimeter packing.}
In some packing problems, the ``container'' has no predefined boundaries (contrary to the cases of strip and bin packing and the study of critical densities), but the pieces can be placed anywhere in the plane and the container is dynamically updated as the bounding box or the convex hull of the pieces.
The goal is then to minimize the size of the container.
In 2D versions of this problem, natural measures of size are the area or the perimeter of the container.
Many papers have been written about offline versions of these problems~\cite{DBLP:conf/stoc/Milenkovic96,MILENKOVIC19993,milenkovich1999translational,ahn2012aligning, althurtado,leewoo,PARK20161,DBLP:journals/eatcs/Alt16,alt_convexOffline_JoCG,lubachevsky2003dense,SPECHT201358,LUBACHEVSKY20091947,erdos1975packing,chung2019efficient}.

Online versions have received relatively little attention.
Fekete and Hoffmann~\cite{fekete2017online} studied online packing axis-parallel squares so as to minimize the area of their bounding square, and gave an $8$-competitive algorithm for the problem.
Abrahamsen and Beretta~\cite{AbrahamsenBeretta20} gave a $6$-competitive algorithm for the same problem and studied the more general case where the pieces are axis-parallel rectangles and we want to minimize the bounding box, with or without rotations by $90^\circ$ allowed.
They gave a $3.98$-competitive algorithm for minimizing the perimeter and showed that there exists no competitive algorithm for minimizing the area, when the pieces can be arbitrary rectangles.

If the pieces are convex polygons that can be arbitrarily rotated, then the minimum perimeter problem can be reduced to the case of packing axis-parallel rectangles by first rotating each piece so that a diameter of the piece is horizontal.
Then the density of the piece in its axis-parallel bounding box is at least $1/2$, and the algorithm for rectangles can be applied to the bounding box.
An interesting question that remained open was therefore whether there is a competitive algorithm for minimizing the perimeter when the pieces are convex polygons that can \emph{not} be rotated.
We answer this question in the negative (\cref{thm:PackingMasterThm}~\ref{MP:d}).

\paragraph{Critical densities and square packing.}
The study of critical densities dates back at least to the 1930s.
In the famous \emph{Scottish Book}~\cite{mauldin2020scottish}, Auerbach, Banach, Mazur and Ulam gave the following theorem (slightly rephrased) and corollary without a proof.

\begin{theorem*}[Potato Sack Theorem, \cite{mauldin2020scottish}]
If $\{K_n\}_{n=1}^\infty$ is a sequence of convex bodies in $\mathbb R^3$, each of diameter $\leq \delta$ and the sum of their volumes is $\leq V$, then there exists a cube with sidelength $s=f (\delta,V)$ such that all the given bodies can disjointly be placed into it when rotations are allowed.
\end{theorem*}

\begin{corollary*}
One kilogram of potatoes can be put into a finite sack.
\end{corollary*}

A simple proof of the theorem, generalized to an arbitrary dimension, was given by Kosi{\'n}ski~\cite{kosinski1957proof}.
It was asked in~\cite{mauldin2020scottish} to determine the function $f (\delta,V)$, which, in modern terms, means to determine the \emph{critical density}.
That is, to find the largest value of $V$ such that a sequence of convex bodies of diameters at most~$\delta$ and total volume $V$ can always be placed in the unit cube.
This theorem has sprouted a lot of interest in determining critical densities in various settings.
For instance, Moon and Moser~\cite{moon1967some} proved that any sequence of $d$-dimensional cubes of total volume $1/2^{d-1}$ can be packed into the unit cube.
As two cubes with sidelengths $1/2+\varepsilon$, for any $\varepsilon>0$, cannot be packed in the unit cube, this shows that the critical density of packing cubes into a unit cube is $1/2^{d-1}$ for any $d\geq 1$.
Alt, Cheong, Park, and Scharf~\cite{alt2019packing} showed that there exist $n$ 2D unit disks embedded in 3D (with different normal vectors) such that whenever they are placed in a non-overlapping way, their bounding box has volume $\Omega(\sqrt n)$.
It follows that when rotations are not allowed, the critical density of packing convex bodies of bounded diameter into a cube is $0$, or, in other words, that one kilogram of potatoes cannot always be put into a finite sack by translation.
In contrast to this, the critical density of packing convex polygons of bounded diameter into the unit square by translation is positive when the diameter is sufficiently small, as we prove in \cref{thm:offline}~\ref{OP:c}.

The study of critical densities likewise makes sense when the pieces appear in an online fashion.
A lower bound on the critical density of online packing squares into the unit square has been improved in a sequence of papers~\cite{januszewski1997line,Brubach2014improved,han2008online,fekete2017online} from $5/16$~\cite{januszewski1997line} to $2/5$~\cite{Brubach2014improved}.
Interestingly, Januszewski and Lassak~\cite{januszewski1997line} proved that in dimension $d\geq 5$, the critical density of online packing cubes into the unit cube is $1/2^{d-1}$, just as in the offline case.

Lassak and Zhang~\cite{lassak1991line} proved that the Potato Sack Theorem also holds for any dimension $d\geq 1$ when the convex bodies appear online, if rotations are allowed. In order to achieve this, each convex body of volume $V$ is rotated so that it has an axis-parallel bounding box of volume at most $d!\cdot V$.
The problem is therefore reduced to online packing axis-parallel boxes.
In simplified terms, it is then proved that for some constant $\delta=\delta(d)>0$, any sequence of axis-parallel boxes of diameter and total area at most $\delta$ can be packed online in the $d$-dimensional unit hypercube.
Determining whether the critical density of translational \emph{and} online packing convex 2D polygons is positive remained an interesting question:
On one hand, this packing problem is harder than the 2D offline version which has positive critical density (\cref{thm:offline}~\ref{OP:c}), and on the other hand, it is easier than the 3D online version which has $0$ critical density (since also the 3D offline version has $0$ critical density~\cite{alt2019packing}).
In this paper, we prove that the 2D \emph{online} version also has critical density~$0$ (\cref{thm:PackingMasterThm}~\ref{MP:c}).
\Cref{table:packing} gives a brief overview of how the result relates to the existing results.

\begin{table}
\centering
\begin{tabular}{|l|l|l|l|l|l|l|}
\hline
Objects & Dimension & Online/offline & Rotation allowed? & Critical density & Reference \\
\hline\hline
Convex & $d\geq 1$ & -- & Yes & $>0$ & \cite{lassak1991line} \\
Axis-parallel boxes & $d\geq 1$ & -- & -- & $>0$ & \cite{lassak1991line} \\
Convex & $3$ & -- & No & $0$ & \cite{alt2019packing} \\
Convex & $2$ & Online & No & $0$ & \Cref{thm:PackingMasterThm} \ref{MP:d} \\
Convex & $2$ & Offline & -- & $>0$ & \Cref{thm:offline} \ref{OP:c} \\
\hline
\end{tabular}
\caption{Results on the critical densities for packing convex objects of diameter at most some sufficiently small constant $\delta>0$ in a unit cube.
A dash means that the result holds regardless.}
\label{table:packing}
\end{table}

\paragraph{Packing irregular pieces.}
Compared to the amount of existing work on packing axis-parallel rectangles or cuboids, the theoretical algorithmic work on packing more general objects is still quite scarce. 
In addition to the previously mentioned papers~\cite{alt_convexOffline_JoCG,alt_convexOffline_corr,alt2019packing,lassak1991line,DBLP:conf/stoc/Milenkovic96,DBLP:journals/algorithmica/Milenkovic97,milenkovich1999translational,DBLP:journals/algorithmica/DanielsM97,MILENKOVIC19993}, we refer the reader to the work on the knapsack problem for convex polygons under rigid motions by Merino and Wiese~\cite{DBLP:conf/icalp/MerinoW20}, the $\exists\mathbb R$-hardness of many variants of offline packing problems by Abrahamsen, Miltzow and Seiferth~\cite{DBLP:conf/focs/AbrahamsenMS20}, and the approximation algorithms for offline packing convex polyhedra in three dimensions under rigid motions by Alt and Scharf~\cite{DBLP:journals/ijcga/AltS18}.

The scarcity of theoretical work is in contrast to the staggering amount of work from the more applied angle of operations research, with numerous papers describing and evaluating heuristics for packing irregular objects.
We refer to some surveys for an overview~\cite{bennell2008geometry, bennell2009tutorial, dyckhoff1992cutting, hopper2001review,
leao2020irregular, sweeney1992cutting}.

\subsection{Online sorting, TSP and scheduling}\label{sec:tsp}

\cref{thm:SortingUpperBoundBaby} can be seen as an asymptotically tight analysis of the traveling salesperson problem (TSP) on the real line, following the \emph{online-list} paradigm:
We want to visit $n$ cities in the unit interval $[0,1]$ over the course of $n$ days, one city per day. The positions of the cities are revealed sequentially to us in an online fashion, and for each city, we have to immediately decide which day to visit that city. Our goal is to minimize the total distance travelled. In fact, we could equally well imagine that we had $\gamma n$ days for our tour which corresponds to \OnlineSorting{}$[\gamma,n]$.

The usually studied version of online TSP follows the \emph{online-time} paradigm.
Here, a \emph{server} starts at point $0$ at time $0$ and moves with at most unit speed.
Each request $\sigma_i=(t_i,r_i)$ is revealed at some time $t_i$ and should be visited by the server at time $r_i$ or later.
We want to minimize the time before the server has visited all requests $\sigma_1,\ldots,\sigma_n$.
This problem has been intensely studied~\cite{bjelde2021tight}. 
The distinction between these two paradigms is common in the area of scheduling, but the online-list variant of TSP has apparently not received any attention so far.

Fiat and Woeginger~\cite{fiat1999line} studied a scheduling problem following the online-list paradigm that seems particularly related to online sorting:
The goal is to minimize the average job completion time in a system with $n$ jobs and a single machine.
In every step, a single new job arrives and must be scheduled to its time slot immediately and irrevocably and without knowledge of the jobs that arrive in later steps.
The offline optimum is to schedule the jobs according to their processing times in sorted order.
It was shown that no algorithm can be $\log n$-competitive, but that there is a $O(\log^{1+\varepsilon} n)$-competitive algorithm for all $\varepsilon>0$.
For surveys on (online) scheduling, see~\cite{sgall1998online,pruhs2004online,graham1979optimization,chen1998review,albersOnlineScheduling}.

\subsection{Structure of the paper}

In \cref{sec:terminology}, we introduce our terminology and notation.
In \cref{sec:problems}, we describe the connection between the problems of online sorting and online strip packing.
In \cref{sec:onlinesorting}, we analyze the online sorting problem and prove \cref{thm:TheLowerBound,thm:SortingUpperBoundBaby,thm:SortingUpperBound}.
In \cref{sec:onlinepacking}, we study online packing problems and present proofs for \cref{thm:PackingMasterThm,thm:PackingUpperBound}.
Finally, in \cref{sec:offlinepacking}, we consider the offline versions of the packing problems and prove \cref{thm:offline}.
See \cref{fig:reductions} for an overview of the reductions we make.

\begin{figure}[htb]
\centering
\includegraphics[page=15,width=\textwidth]{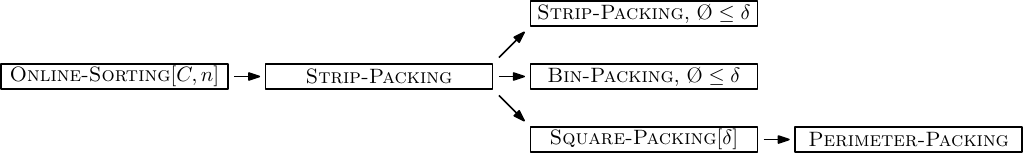}
\caption{An overview of our reductions.
Note that an arrow from problem $A$ to problem $B$ means that an algorithm for $B$ implies an algorithm for $A$.
Here, $Ø \leq\delta$ means that the diameter of each piece is at most an arbitrary constant~$\delta>0$.
}
\label{fig:reductions}
\end{figure}

\section{Terminology}\label{sec:terminology}

In the online problems studied in this paper, the input is a stream $\sigma_1,\ldots,\sigma_n$ of objects, and we need to handle each object $\sigma_i$ before we know the next ones $\sigma_{i+1},\ldots,\sigma_n$.
Here, objects are either real numbers from the unit interval $[0,1]$, in which case we call them \emph{reals}, or they are convex polygons, in which case we call them \emph{pieces}.
The problems will be defined in more detail in Section~\ref{sec:problems}.

Let us now revisit the standard terminology of competitive analysis for an online algorithm $\A$ of a minimization problem $\mathcal P$.
For any instance $I$ of $\mathcal P$, let $\OPT(I)$ denote the cost of the offline optimum solution and $\A(I)$ denote the cost of the solution that $\A$ produces on input $I$.
Let $f$ be a function from the set of instances to the real numbers (the functions $f$ we consider will in fact only depend on $n$, the number of pieces in $I$).
We say that $\A$ has \emph{(absolute) competitive ratio} $f(I)$ if for all instances $I$ it holds that
\[
\A(I)\leq f(I)\cdot \OPT(I).
\]
If $\A$ has competitive ratio $f(I)\leq c$ for some constant $c$, then we say that $\A$ is \emph{competitive}.

Similarly, we say that $\A$ has \emph{asymptotic competitive ratio} $f(I)$ if there exists $\beta>0$ such that for all instances $I$ it holds that
\[
\A(I)\leq f(I)\cdot \OPT(I)+\beta.
\]

For a point $p\in\mathbb R^2$, we denote by $x(p)$ and $y(p)$ the $x$- and $y$-coordinates of $p$, respectively.
For a compact set of points $S\subset\mathbb R^2$, we define the \emph{diameter} of $S$ as $\max_{p,q\in S} |p-q|$.
We furthermore define the \emph{width} and \emph{height} of $S$ as
\[\pwidth(S)\mydef \max_{p,q\in S} |x(p)-x(q)|\quad\text{and} \quad  \pheight(S)\mydef \max_{p,q\in S} |y(p)-y(q)|.\]
If $S$ is a polygon, we denote the area of $S$ as $\area(S)$.

A parallelogram $P$ is \emph{horizontal} if $P$ has a pair of horizontal edges.
The horizontal edges of $P$ are then called the \emph{base} edges.
The \emph{shear} of a horizontal parallelogram $P$ is $x(t)-x(b)$, where $t$ and $b$ are the top and bottom endpoints of a non-horizontal edge of $P$, respectively.

\section{The connection between online sorting and packing}\label{sec:problems}

In (translational) \StripPackingConv, we have a horizontal strip $S$ of height 1 which is bounded to the left by a vertical segment and unbounded to the right; see \cref{fig:problems}.
We have to pack a stream of convex polygonal pieces appearing online.
We must place each piece before we know the next.
Specifically, each piece must be placed in $S$ by a translation such that it is interior disjoint from all previously placed pieces.
The \emph{occupied} part of $S$ is the part from the left end of $S$ to the vertical line through the rightmost corner of a piece placed in $S$.
The \emph{width} or the \emph{cost} of a packing is the width of the part of $S$ that is occupied, i.e., the horizontal distance from the left end of $S$ to the rightmost corner of a piece placed in $S$.

\begin{figure}[htb]
\centering
\includegraphics[page=1]{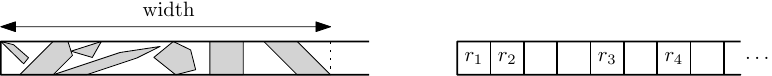}
\caption{Strip packing and online sorting.}
\label{fig:problems}
\end{figure}

In the problem \OnlineSorting{}$[\gamma,n]$, we are given an array $A$ with $\lfloor \gamma n\rfloor$ cells.
Each arriving real is contained in the unit interval $[0,1]$.
After all $n$ reals have been placed, define $\mathbf r\mydef (r_1,\dots, r_n)$ to be the numbers according to their left-to-right order in $A$.
Furthermore, define the sentinel values $r_0\mydef 0$ and $r_{n+1}\mydef 1$.
Then the cost is given by
\[\cost(\mathbf r)\mydef \sum_{i=0}^{n} |r_{i+1}-r_i|.\]
The offline optimum is achieved when the reals are in sorted order and is then exactly $1$.

\cref{fig:reductions} gives an overview of reductions we make.
Arguably, the crucial reduction is that from \OnlineSorting{} to \StripPackingConv{}, as described by the following lemma.

\begin{lemma}\label{lemma:stripToArrayComplete}
If there exists a $C$-competitive algorithm for \StripPackingConv{}, then there exists a $4C$-competitive algorithm for \OnlineSorting{}$[2C,n]$.
\end{lemma}

\begin{proof}
Suppose that we have a $C$-competitive algorithm $\A_1$ for \StripPackingConv{}.
Let $s_1,\ldots,s_n$, $s_i\in[0,1]$, be a stream of reals that we wish to sort online in an array $A$ of size $\lfloor 2Cn\rfloor$.
For each real $s_i$, we construct a parallelogram $P_i$ with height $1$, base edges of length $1/n$, and shear $s_i$.
We then present $\A_1$ with $P_i$ and observe where the bottom left corner of $P_i$ is placed in the strip.
Let $x_i$ be the $x$-coordinate of this corner (suppose that the line $x=0$ forms the left boundary of the strip).
We then place $s_i$ in the cell of index $\lfloor nx_i\rfloor$ in the array $A$, and denote the resulting algorithm $\A_2$.
Since the base segments of the parallelograms have length $1/n$, this will not cause any collisions in $A$.

By sorting the pieces $P_1,\ldots,P_n$ in order of increasing shear and placing them in this order in the strip, we obtain a packing of width at most $2$, so $\mathcal{A}_1$ will place all pieces within a prefix of size $2C\times 1$ of the strip.
Hence, $\mathcal{A}_2$ will place each real in $A$.

\begin{figure}[htb]
\centering
\includegraphics[page=3]{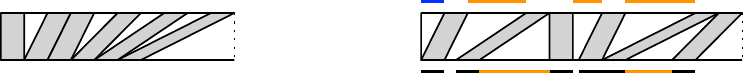}
\caption{
Two strip packings of the same set of parallelograms.
Left: The packing we use as an upper bound for the optimum (which is in fact the optimum).
Right: An arbitrary packing with the gaps and segments that we count shown as segments outside the strip.
}
\label{fig:packingsorting}
\end{figure}

Let $\mathbf r \mydef (r_1,r_2,\ldots,r_n)$ be the numbers in the resulting left-to-right order in $A$ produced by $\A_2$, and let $P'_i$ be the parallelogram corresponding to $r_i$.
We have the following contributions to the width of the resulting packing; see also \cref{fig:packingsorting}.
\begin{itemize}
\item Between the vertical left boundary of the strip and the top edge of $P'_1$, there is a gap of length at least $r_1=|r_1-r_0|$.
\item If $r_i\leq r_{i+1}$, then there is a gap between the top edges of $P'_i$ and $P'_{i+1}$ of length at least $|r_{i+1}-r_i|$.
\item If $r_i> r_{i+1}$, then there is a gap of length at least $|r_{i+1}-r_i|$ between the bottom edges of $P'_i$ and $P'_{i+1}$.
\item The bottom base edges of the pieces have total length $1\geq |r_{n+1}-r_n|$.
\end{itemize}
The sum of all these gaps is at least $\cost(\mathbf r)$, and as half of the sum appear as distances along the top or the bottom edge bounding the strip, we get that the width of the produced packing is at least $\cost(\mathbf r)/2$.
Now, since $\A_1$ is $C$-competitive, we get that $\cost(\mathbf r)/2\leq 2C$, and hence $\cost(\mathbf r)\leq 4C$.
\end{proof}

\section{Online sorting}\label{sec:onlinesorting}
In this section, we present upper and lower bounds for the online sorting problem.
As a warm up, we consider in \cref{sec:baby} the case where we have no extra space, i.e., we are given $n$ reals in $[0,1]$ in an online fashion to be inserted into an array of length $n$.
We prove \cref{thm:SortingUpperBoundBaby}, which gives an asymptotically tight analysis of the optimal competitive ratio in this case.
In \cref{sec:onlineLower}, we proceed to prove \cref{thm:TheLowerBound}, thereby obtaining a general lower bound on the competitive ratio for \OnlineSorting{}$[\gamma,n]$.
Finally, in \cref{sec:upperonline}, we give an upper bound on the competitive ratio of \OnlineSorting{}$[\gamma,n]$ for $\gamma>1$.

\subsection{Tight analysis of \sc{Online-Sorting}\texorpdfstring{$[1,n]$}{1n}}\label{sec:baby}

For online sorting $n$ numbers in an array $A$ of size $n$, we can prove asymptotically tight bounds on the optimal competitive ratio.\ifsa\footnote{We thank Shyam Narayanan for improving our initial $O(\sqrt{n}\log n)$ upper bound to the asymptotically tight $O(\sqrt{n})$.}\fi\ 
We restate~\cref{thm:SortingUpperBoundBaby} below.

\SortingUpperBoundBaby*

\begin{proof}
We start out by proving the lower bound.
Let $N\mydef \lfloor\sqrt{2n} \rfloor$. Consider a fixed but arbitrary online algorithm $\mathcal A$. 
We are taking the role of the adversary and present reals of the form $\nicefrac{k}{N}$ with $k=0,1,\dots, N$.
Clearly, if all reals are placed in an increasing fashion, the resulting cost is 1. This is the value $\mathcal A$ has to compete against.

Recall that we defined the sentinel values $r_0=0$ and $r_{n+1}=1$. Equivalently, we will think of the array as being appended with a cell to the left and to the right initially containing the reals $0$ and $1$. We want to repeatedly pick a real so that there are no adjacent duplicates in the array immediately after the algorithm places our real.
To this end, we define a real to be \emph{expensive} if it does not appear in a cell with an empty adjacent cell.
Our strategy is as follows.
While there exists one, we present an expensive real.
If we are able to present $n$ expensive reals, then any two consecutive entries are distinct and differ by at least $\nicefrac{1}{N}$. Consequently, the cost is at least~$\frac{n+1}{N}\geq\sqrt{n/2}$; recall that we have an additional sentinel entry at each end of the array. 
On the other hand, if we reach a point where no real of the form $\nicefrac{k}{N}$ is expensive, we present $0$'s until all entries are filled. Then for any $k$, a real of the form $\nicefrac{k}{N}$ will appear next to a $0$ in the array, so the total cost is at least $\sum_{k=0}^N \nicefrac{k}{N}
= (N+1)/2> \sqrt{n/2}$.
\medskip

Next we describe an algorithm $\mathcal{A}$ for \OnlineSorting{}$[1,n]$ with competitive ratio at most $18\sqrt{n}$.
Define $N_1\mydef \lfloor\sqrt{n}\rfloor$ and partition $[0,1]$ into $N_1$ intervals, $J_1,\dots,J_{N_1}$, each of length $\nicefrac{1}{N_1}$.
Further define $N_2\mydef 2N_1$ and partition the array $A$ into $N_2$ subarrays of contiguous cells, $A_1,\dots,A_{N_2}$, each of size either $\lfloor n/N_2 \rfloor$ or $\lceil n/N_2 \rceil$. Whenever $\mathcal{A}$ receives a real $x\in J_i$, it checks whether there exists a $j$ such that $A_j$ is not full and already contains a number from $J_i$. If so, $\mathcal{A}$ places $x$ in the leftmost empty cell of $A_j$.
Otherwise,
$\mathcal{A}$ checks if there exists a $j$ such that $A_j$ is empty. If so, $\mathcal{A}$ places $x$ in the leftmost cell of such an $A_j$. 
Finally, if neither of the two above cases occur, then
$\mathcal{A}$ recurses on the subarray formed by the union of the empty cells of $A$.

Let $\mathbf r=(r_1,\ldots,r_n)$ be the numbers in the resulting left-to-right order in $A$.
Define $\cost'(\mathbf r)\mydef\sum_{i=1}^{n-1} |r_{i+1}-r_i|$, and observe that $\cost'(\mathbf r)=\cost(\mathbf r)-|r_1-r_0|-|r_{n+1}-r_n|$, i.e., $\cost'(\mathbf r)$ equals the full cost when we ignore the sentinel values $r_0$ and $r_{n+1}$.
We prove by induction on $n$ that $\cost'(\mathbf r)$ never exceeds $T(n)=\alpha \sqrt{n}$, where $\alpha\mydef \frac{5\sqrt{2}}{\sqrt{2}-1}\approx 17.1$.
This clearly holds for $n\leq \alpha^2$, as then $T(n)\geq n$ which is a trivial upper bound on $\cost'(\mathbf r)$.
So let $n>\alpha^2$ and assume inductively that the result holds for arrays of length $n'<n$.
Let $B\subset\{1,\dots,n\}$ denote the set of indices $i$ such that $A[i]$ is full when the algorithm recurses for the first time. At this point in time, all the subarrays $A_1,\dots,A_{N_2}$ are either filled or partially filled since otherwise, the algorithm would not recurse. Moreover, each partially filled subarray contains reals coming exclusively from a single interval $J_i$ and no two partially filled subarrays contain reals from the same interval $J_i$. Thus, at this point at most $N_1-1$ of the $N_2= 2N_1$ subarrays $A_1,\dots,A_{N_2}$ are not completely filled, and each subarray must contain at least one real. In particular $|B| \geq (N_1+1)\lfloor n/(2N_1)\rfloor + (N_1-1) \geq n/2$. 
Let $B^c\mydef \{1,\dots,n\}\setminus B$ and define $D_1\mydef \{1\leq i <n \mid \{i,i+1\}\subset B^c\}$, and $D_2\mydef \{1,\dots,n-1\}\setminus D_1$.
Define $T_j \mydef \sum_{i\in D_j}|A[i+1]-A[i]|$ for $j\in\{1,2\}$.
We then have $\cost'(\mathbf r)\leq T_1+T_2$.
By the induction hypothesis, $T_1\leq T(|B^c|)\leq T(\lfloor n/2 \rfloor)$.
Furthermore, we can bound $T_2$ from above as
\[
T_2\leq \frac{n}{N_1}+\frac{3N_2}{2}.
\]
Here the first term comes from consecutive entries internal to a subarray that are both filled before the algorithm recurses for the first time (these have pairwise distance at most $1/N_1$).
The second term comes from consecutive entries $A[i],A[i+1]$ where $A[i]$ is in some subarray $A_j$ and $A[i+1]$ is in the next $A_{j+1}$, or $A[i]$ is full when the algorithm recurses for the first time and $A[i+1]$ becomes full in the recursion (in which case we simply bound $|A[i]-A[i+1]|\leq 1$). 
We therefore have
\[
\cost'(\mathbf r)\leq T_1+T_2\leq T(\lfloor n/2 \rfloor)+ \frac{n}{N_1}+\frac{3N_2}{2}\leq \alpha\sqrt{\frac{n}{2}}+5\sqrt{n}=\alpha \sqrt{n}.
\]
It is easy to check that including the sentinel values, the total cost incurred by the algorithm never exceeds
$18\sqrt{n}$.
This completes the proof.
\end{proof}

Having dealt with the case where we have no extra space, we turn to the setting where the array has length $\gamma n$ for some  $\gamma>1$. This is the setting which is important for our reductions to the online packing problems. In the following two sections, we prove lower and upper bounds, respectively, on how good online sorting algorithms can perform in this case.

\subsection{Lower bound for the general case of \sc{Online-Sorting}}\label{sec:onlineLower}
Let us restate our lower bound for the competitive ratio of any algorithm for  \OnlineSorting{}$[\gamma,n]$, $\gamma\geq 1$.

\TheLowerBound*

Let us remark that the adversarial stream of reals leading to the lower bound is chosen in a deterministic and adaptive way, i.e., depending on where the preceding reals have been placed in the array.
Note that a deterministic oblivious adversary cannot lead to a lower bound above $1$ on $\Delta$ since for any fixed stream of reals, there exists an algorithm that places them in sorted order.
It is however an interesting question whether there exists a \emph{randomized} oblivious adversary matching or exceeding our lower bound.

Before delving into the proof, let us also describe the high level idea on how to generate an adversarial stream that incurs a high cost for any given algorithm.
Assume for simplicity that $\gamma$ is a constant, e.g., $\gamma=2$ and that we want to prove a lower bound of $\Delta$ on the total cost for some $\Delta =(\log n)^{\Theta(1)}$.
We will start by presenting the algorithm with reals coming from a set $S$ of the form $S=\left\{i/n_0\mid i \in \{0,\ldots,n_0\}\right\}$ for some $n_0\leq n$. At any point, we consider the set of maximal intervals of empty cells of the array, call them $I_1,\dots, I_\ell$.  For each real $x\in S$, we define the home $H(x)$ as the union of all such intervals $I_j$ such that placing $x$ in $I_j$ incurs no extra cost, i.e.,  the first non-empty cell to the left or right of $I_j$ contains the real~$x$. (In fact, we will be a little more generous in the actual proof and allow for some small extra cost). If $|H(x)|< \frac{n}{\Delta n_0}$ for some $x\in S$, we simply present the algorithm with copies of $x$ until one is placed outside $H(x)$ and thus has distance at least $1/n_0$ to its neighbours. This will essentially contribute a total cost of $1/n_0$ to the objective function, i.e., an average cost of $\Delta/n$ per presented copy of~$x$. Note that this is the correct average cost for a lower bound of~$\Delta$. 
However, it may well be the case that no such $x\in S$ exists (the average size of $H(x)$ for $x\in S$ is $\approx \gamma n /n_0$ which is much larger). In this case, we \emph{coarsen} the set $S$ to a set $S'\subset S$ consisting of every $s$'th element of $S$ for some $s=\text{polylog} \, n$ and only present the algorithm with reals from $S'$ from this point on. 
Now for most $x\in S$, it holds that $H(x)=O(n/n_0)$ and that the distance from $x$ to any real in $S'$ is $\Omega(s/n_0)$. Intuitively, this means that filling up $H(x)$ with elements from $S'$ has a high cost of $s/n\gg \Delta/n$ per presented real. We prove that this implies that we can point to a `deserted space' consisting of $\Omega(n/\Delta)$ empty cells, in which the algorithm can only place a negligible number of  reals without incurring a large total cost of $\Delta$.
Now we continue the process starting with $S'$.
In each coarsening step, we specify
a `deserted space' of size $\Omega(n/\Delta)$ and we can importantly enforce that these spaces be disjoint.
As the array has $2n$ cells in total, this coarsening can happen at most $O(\Delta)$ times. To ensure that we can in fact perform this coarsening $\Omega(\Delta)$ times, we must ensure that $s^\Delta \ll n$, which in turn implies that $\Delta =O(\log n /\log \log n)$. 

\begin{proof}[Proof of \cref{thm:TheLowerBound}]
Let $\A$ denote any online algorithm for \OnlineSorting{}$[\gamma,n]$.
We may assume that $n$ is sufficiently large and that $\gamma\leq \log n /\log \log n$.
We will present the algorithm with an (adaptive) stream that incurs a cost of $\Omega \left({\log n}/{\gamma \log \log n}\right)$. 

Let $C\in [3,4]$ be minimal such that $s\mydef \log^C n$ is an integer and define $\delta\mydef \frac{\log n}{16C \gamma \log \log n}$.

For $i\in \mathbb{N}$, we define 
\[
S_i\mydef \left\{k\cdot \frac{s^i}{n}\,\bigg|\, k \in \{0,\ldots,\lfloor n/s^i\rfloor\}
\right\}
\]
such that $S_1\supset S_2 \supset \cdots$.
We also let $i^*\in \mathbb{N}$ be maximal with $s^{i^*}\leq n$.
In other words, we define $i^*\mydef \lfloor {\log n}/{C\log \log n}\rfloor$.

Our adversarial stream will consist of several phases, where in phase $i\geq 1$, we present $\mathcal{A}$ with reals from $S_i$. If at some point we have presented $\mathcal{A}$ with $n'$ reals and are at some phase $i$, we will say that we are at \emph{time} $t=n'+i$. For each step of the construction of our adversarial stream, the time will increase by one, either because we presented $\mathcal{A}$ with a real, or because we changed to phase $i+1$ from some phase $i$.

By the end of each phase, we will mark a certain set of currently empty cells of the array. The marked cells will represent parts of the array where $\mathcal{A}$ can only place a limited number of
reals
without incurring a high cost. Identifying the array with $[m]$, at any given point in time $t$,
we let $F_t\subset [m]$ denote the full cells, $M_t\subset [m]$ denote the marked cells, and $R_t\mydef [m] \setminus (F_t\cup M_t)$ denote the remaining cells. We remark that $F_t$ and $M_t$ need not be disjoint; even though a cell can only get marked when it is empty, $\mathcal{A}$ might insert a real in that cell at a later point in time. Initially, we let $M_0=\emptyset$, so that $R_0=[m]$ is the entire array.

Suppose that we are at time $t$ and in some phase~$i$. For an empty cell $p\in R_t$, we define $N_t(p)\subset [m]$ to be the set consisting of the first non-empty cell to the left and to the right of $p$. Thus $|N_t(p)|\leq 2$. Furthermore, for $x\in S_i$, we define the \emph{home} of $x$ at time $t$ in phase $i$ to be the set
\[
H_t^{i}(x)\mydef \left\{p\in R_t\,\bigg|\, \exists q\in N_t(p) \text{ such that } |A[q]-x|< \frac{s^i}{2 n} \right\}.
\]
Note that by definition, the home is contained in $R_t$ and thus consists of cells that are unmarked and empty at time $t$. By construction, a cell of $A$ can be contained in at most two homes. We say that the home of $x$ is \emph{small} at time $t$ if $|H_t^{i}(x)|<\frac{s^i}{\delta}$. 
In this case, we say that $x$ is \emph{expensive} at time $t$.

The adversarial stream is generated through a sequence of phases as described below, where we terminate as soon as the array contains exactly $n$ reals.

\medskip

\noindent\textbf{Phase $0$:}

\begin{itemize}[label={},noitemsep,nolistsep]
\item
Present $\A$ with
any real from $S_1$.
Proceed to phase $1$.
\end{itemize}

\noindent\textbf{Phase $i\geq 1$:}

\begin{itemize}[label={},noitemsep,nolistsep]

\item
While there exists an expensive real $x \in S_i$ at some time $t=t_0$: 

\begin{itemize}[noitemsep,nolistsep]
\item
Present $\A$ with copies of $x$ until one is placed in a cell of $R_{t_0}\setminus H_{t_0}(x)$ or $A$ contains $n$ reals.

\item
Terminate if $A$ contains $n$ reals.

\item
Leave the set of marked cells unchanged.
\end{itemize}

\item
If no expensive real exists at the current time $t$: 

\begin{itemize}[noitemsep,nolistsep]
\item
Define the set of \emph{well-sized} reals at phase $i$ as
$
W_i\mydef 
\big\{x\in S_i \,\big|\, \frac{s^i}{\delta} \leq  |H_t^{i}(x)| \leq 4\gamma s^i
\big\}.
$

\item
Define the set of \emph{deserted reals} at phase $i$ as
$
D_i\mydef \big\{x\in W_i\,\big|\, \forall y \in S_{i+1} \text{ it holds that } |y-x|\geq \frac{s^{i+1}}{12 n}\big\}.
$

\item
Define the \emph{deserted space} at phase $i$ as $\mathcal{D}_i\mydef \bigcup_{x\in D_i}H_t^{i}(x)$.

\item
Mark all the cells in $\mathcal{D}_i$ (so that $M_{t+1}\mydef M_{t}\cup \mathcal{D}_i$).
\end{itemize}

\item
Proceed to phase $i+1$.
\end{itemize}

\medskip

Note that ${s^{i^*+2}}/{\delta}\geq {sn}/{\delta}>\gamma n$, so in phase $i^*+2$, the real $0$ is always expensive. In particular, the process will stop by phase $i^*+2$, if not earlier. 
We remark for later use that the sets $\mathcal{D}_i$ are pairwise disjoint. Indeed, $\mathcal{D}_i$ is a union of homes at time $t$, and therefore consists of currently unmarked cells. However, the cells of  $\mathcal{D}_i$ are then immediately marked.

We prove three lemmas below which combine to give our desired result.
\begin{lemma}\label{lemma:marked}
Let $i\leq i^*-2$. If the algorithm does not terminate in phase $i$, then $|\mathcal{D}_i|\geq \frac{n}{8 \delta}$.
\end{lemma}

\begin{proof}
We consider the situation in the end of a phase $i$ where the algorithm has not terminated.
In this case, for each $x\in S_i\setminus W_i$, we have $\vert H_t^{i}(x)\vert >4\gamma s^i$. In particular, the number of pairs $(x,p)$ such that $x\in S_i\setminus W_i$ and $p\in H_t^{i}(x)$ can be lower bounded by $4\gamma s^i |S_i\setminus W_i|$. On the other hand a cell can lie in at most two homes, so the number of such $(x,p)$ pairs is upper bounded by 
\[
\sum_{p\in [\gamma n]}|\{x\in S_i: p \in H_t^{i}(x)\}|\leq \sum_{p\in [\gamma n]}2=2\gamma n.
\]
It follows that 
\[
|S_i\setminus W_i|\leq   \frac{ n}{2 s^i}\leq \frac{1}{2}\left(\left\lfloor \frac{n}{s^i} \right\rfloor+1\right)= \frac{|S_i|}{2},
\]
so that $|W_i|\geq \frac{|S_i|}{2}$.
Now $S_{i+1}$ can be obtained from $S_i$ by including every $s$'th elements from $S_i$ in increasing order starting with $0$, and thus, $|S_{i+1}|=\lceil |S_{i}|/s \rceil$.
Let $D_i'\mydef \left\{x\in S_i\mid \forall z \in S_{i+1} \text{ it holds that } |z-x|\geq \frac{s^{i+1}}{12 n}\right
\}$. Then 
\[
|S_i\setminus D_i'|\leq |S_{i+1}|\left(2\left\lfloor \frac{s}{12}\right\rfloor +1\right) \leq |S_{i}|\left(\frac{1}{6}+\frac{1}{s}\right)\left(1+\frac{s}{|S_i|}\right).
\]
By the definition of $i^*$ and using that that $i\leq i^*-2$, we have that $|S_i|\geq \frac{n}{s^i}\geq s^2$. Using the assumption that $n$ and hence $s$ is sufficiently large, it thus follows that $|S_i\setminus D_i'|\leq \frac{|S_i|}{4}$ and thus, $|D_i'|\geq \frac{3|S_i|}{4}$. Combining the lower bounds on $|W_i|$ and $|D_i|$, we obtain that 
\[
|D_i|=|W_i\cap D_i'|=|W_i|+|D_i|-|W_i\cup D_i'|\geq \frac{|S_i|}{2}+\frac{3|S_i|}{4}-|S_i|=\frac{|S_i|}{4}
\]
Again using that each cell is contained in at most two homes, we have that 
\[|\mathcal{D}_i|\geq \frac{1}{2} \cdot\frac{|S_i|}{4}\cdot \frac{s^i}{\delta}\geq \frac{n}{8 \delta}.\qedhere\]
\end{proof}
\begin{lemma}\label{lemma:marked2}
Let $\alpha>0$ and assume that the algorithm does not terminate in phase $i$. If at least $56{\alpha n \gamma}/{s}$ reals from $S_{i+1}$ are placed in $\mathcal{D}_i$, then the total cost is at least $\alpha$.
\end{lemma}
\begin{proof}
Write $\mathcal{D}_i=\bigcup_{\ell}J_\ell$ as a disjoint union of maximal intervals of empty cells in $\D_i$.
For each $\ell$, we have that $|J_\ell|\leq 4\gamma s^i$.
If at least $56 {\alpha n \gamma}/{s}$ reals are placed in $\mathcal{D}_i$, it follows that at least ${14\alpha n}/{s^{i+1}}$ of the intervals receive at least one real from $S_{i+1}$.
For such an interval $J_\ell$, we have by the definition of a home, that one of the (up to two) non-empty cells immediately to the left and right of $J_\ell$ contains a real $x$ of distance at most $\frac{s^i}{2n}$ to an element of~$D_i$.
However, any $x'\in S_{i+1}$ placed in $J_\ell$ has distance at least $\frac{s^{i+1}}{12n}$ to any element of $D_i$.
It follows that $|x-x'|\geq \frac{s^i}{2n}(\frac{s}{6}-1)\geq \frac{s^i}{2n} \frac{s}{7}= \frac{s^{i+1}}{14n}$, assuming that $n$ and hence $s$ are sufficiently large.
Since the $J_\ell$ are disjoint, it easily follows that the total cost is at least $\frac{s^{i+1}}{14n}\cdot \frac{14 \alpha n}{s^{i+1}}=\alpha$.
\end{proof}

\begin{lemma}\label{lemma:unmarked}
If $\mathcal{A}$ assigns $n_0$ reals to unmarked cells, during phases $1,\dots, i^*$,
then the total cost is at least $\frac{\delta n_0}{4n}-1$.
\end{lemma}
\begin{proof}
Write $a_i$ for the number of reals that the algorithm assigns to unmarked cells during phase $i$, so that $\sum_{i\leq i^*}a_i=n_0$.
Let $b_i\mydef \frac{a_i}{\lceil s^i/\delta\rceil}$ and $c_i\mydef \lfloor b_i \rfloor$.
We make the following claim.
\begin{claim}
The total cost is at least $\sum_{i\leq i^*} c_i\cdot \frac{ s^i}{2n}$.
\end{claim}
\begin{proof}[Proof of Claim]
Note that during the while-loop of our algorithm, the assumption that the home of $x$ is small implies that it must happen at least $c_i$ times during phase $i$ that a real from $S_i$ is placed outside its home.
When an real $x$ gets placed in a cell $p\in R_t\setminus H_t^{i}(x)$ outside its home, the reals stored in the cells in $N_t(p)$ have distance at least $\frac{s^i}{2n}$ to $x$.
We say that an interval of cells $[p_1,p_2]$ form an \emph{$i$-jump} if $|A[p_1]-A[p_2]|\geq \frac{s^i}{2n}$. We say that an $i$-jump $[p_1,p_2]$ and a $j$-jump $[p_1',p_2']$ are disjoint if $(p_1,p_2)\cap(p_1',p_2')=\emptyset$. We will show that we can find a collection $\mathcal{J}$ of  such jumps such that (1) the jumps in $\mathcal{J}$ are pairwise disjoint and (2) we can make a partition $\mathcal{J}=\bigcup_{i\geq 1}\mathcal{J}^{(i)}$ such that $\mathcal{J}^{(i)}$ contains at least $c_i$ $i$-jumps. The claim then follows immediately by the triangle inequality. 

To show the existence of such a collection $\mathcal{J}$, suppose we at step $t$ in some phase $i$ are given a collection of $(\leq i)$-jumps $\mathcal{J}_t$  and a partition $\mathcal{J}_t=\bigcup_{j\leq i}\mathcal{J}_t^{(j)}$ such that $\mathcal{J}_t^{(j)}$ consists of $j$-jumps. 
Suppose further that we  place a real $x$ in a cell $p\in R_t\setminus H_t^{i}(x)$ at step $t$. Since $x$ was placed outside of its home, any of the (up to two) intervals with one endpoint at $p$ and the other at a cell of $N_t(p)$ will form an $i$-jump.
Now if $p$ is not contained in a jump of $\mathcal{J}_t$, we can easily extend to a collection $\mathcal{J}_{t+1}=\bigcup_{j\leq i}\mathcal{J}_{t+1}^{(j)}$ which still satisfies (1) and where $\mathcal{J}_{t+1}^{(i)}$ contains one further $i$-jump; we simply add the $i$-jump between $p$ and either of its neighbours in $N_t(p)$. 
Suppose on the other hand that $p$ is contained in a $j$-jump $[q_1,q_2]$ of $\mathcal{J}_t^{(j)}$ for some $j\leq i$.
In this case, we in particular have that $|N_t(p)|=2$ and we write $N_t(p)=\{p_1,p_2\}$ where $p_1< p_2$. Then $[p_1,p]$ and $[p,p_2]$ are both $i$-jumps and in particular $j$-jumps for $j\leq i$. We then replace $[q_1,q_2]$ with $[p_1,p]$ in $\mathcal{J}_{t+1}^{j}$ and include $[p,p_2]$ in $\mathcal{J}_{t+1}^{i}$ to obtain a collection $\mathcal{J}_{t+1}=\bigcup_{j\leq i}\mathcal{J}_{t+1}^{(j)}$ having the same number of $j$-jumps for $j<i$ but one extra $i$ jump compared to $\mathcal{J}_t$. 
The existence of the collection $\mathcal{J}$ follows immediately from these observations. 
\end{proof}
Now assuming that  $n$ is sufficiently large, we have that $b_i\geq\frac{\delta a_i}{2s^i}$. Moreover, using the definition of $i^*$, it is easy to check that $\sum_{i\leq i^*}s^i\leq 2n$. Combining this with the claim, we can lower bound the total cost by
\[
\sum_{i\leq i^*}c_i\cdot \frac{s^i}{2n}\geq \sum_{i\leq i^*}(b_i-1) \cdot\frac{s^i}{2n}\geq \sum_{i\leq i^*} \frac{\delta a_i}{4n}-\sum_{i\leq i^*} \frac{s^i}{2n}\geq \frac{\delta n_0}{4n}-1,
\]
as desired.
\end{proof}

We now combine the two lemmas to prove a lower bound on the cost incurred by the algorithm on our adversarial stream.
Recall that $C= 3$ and $\delta= \frac{\log n}{16C \gamma \log \log n}$. 
Suppose for contradiction that the algorithm has not terminated by the end of phase $i^*-2$.
Using Lemma~\ref{lemma:marked} and the fact that the sets $(\mathcal{D}_j)_{j\leq i^*-2}$ are disjoint, we obtain that 
\[
\left|\bigcup_{j\leq i^*-2}\mathcal{D}_j\right|\geq (i^*-2)\frac{n}{8\delta} > \frac{\log n}{2C \log \log n} \frac{n}{8\delta}\geq \gamma n,
\]
assuming that $n$ is sufficiently large. This is a contradiction as there are only $\gamma n$ cells in $A$. Thus, the algorithm must terminate during some phase $i_0\leq i^*-2$.
Now the algorithm must either place at least $n/2$ reals in marked cells or $n/2$ reals in unmarked cells. In the former case, there must be a phase $i$ in which $\frac{n}{2i^*}$ reals a placed in $\D_i$, and hence it follows from Lemma~\ref{lemma:marked2} that the total cost is $\Omega(\frac{s}{i^*\gamma})=\Omega(\log n)$.
In the latter case, it follows from \cref{lemma:unmarked} that the total cost is at least $\frac{\delta}{8}-1=\Omega\left( \frac{\log n}{\gamma \log \log n} \right)$.
\end{proof}

\subsection{Upper bound for the general case of \sc{Online-Sorting}}\label{sec:upperonline}
In this section, we design an algorithm that shows the following theorem. 
\SortingUpperBound*

In order to prove \Cref{thm:SortingUpperBound}, we present an algorithm that uses a divide and conquer approach. We divide the set of reals $[0, 1]$ into smaller intervals, and place reals from the same interval in some given sub-arrays that we call boxes. In order to place reals from a given interval into the proper box, we invoke our algorithm recursively. 
In \Cref{lem:foranyk}, we parameterize our algorithm by its recursion depth, then we choose the optimal recursion depth (given the parameters $n$ and $\varepsilon$) in order to obtain \Cref{thm:SortingUpperBound}.

\begin{lemma} \label{lem:foranyk}
Given $\delta \in (0, 1/2)$, $[\alpha, \alpha + \beta) \subseteq [0, 1]$ and $k\geq 1$ with $k\leq 1/(2\delta)+1$, there exists an algorithm that solves $\OnlineSorting \ld[1+2k\delta, n\rd]$ over any stream of reals $r_1, \dots, r_n\in[\alpha, \alpha + \beta)$ achieving a cost of $\beta \cdot n^{1/(k+1)}  \delta^{-O(k+1)}$. 
\end{lemma}

\begin{proof}
We prove the statement by induction on $k$. The base case of $k=1$ follows directly from \cref{thm:SortingUpperBoundBaby}: in fact, it is sufficient to apply the mapping $x \mapsto \alpha + \beta x$ to the reals in our stream and notice that the resulting cost also shrinks by a factor $\beta$. We call this version of the algorithm $\onlinesorter{1}$.

For the induction step, we define the algorithm $\onlinesorter{k}$ using $\onlinesorter{k-1}$ as a subroutine.
Let $b \mydef \lfloor n^{1/(k+1)} \rfloor$, $n'\mydef \lfloor \delta n^{k/(k+1)} \rfloor$ and $w \mydef \lfloor (1+2(k-1)\delta) \cdot n'\rfloor$.
For $i \geq 1$, we define the box 
\[
B_i \mydef \ld[(i-1)\cdot w, i \cdot w\rd).
\]
We define the pointer vector $p:[b] \longrightarrow \mathbb{N}$ and initialize $p(i)=i$ for each $i \in [b]$. We define the set of active boxes $\A \mydef \{B_{p(i)} \,|\, i \in [b]\}$ and let  $\R \mydef \max_{i\in[b]} p(i) \cdot w$ denote the rightmost cell index in any active box.
We say $B_{p(i)}$ is \emph{full} if it contains $n'$ reals.
Given a real $x \in [\alpha, \alpha + \beta)$, let $i \in [b]$ be such that $(i-1) \cdot \beta \leq (x-\alpha) b < i \cdot \beta$. If $B_{p(i)}$ is not full, we place $x$ into $B_{p(i)}$ using $\onlinesorter{k-1}$. Otherwise, we assign a new active box and set $p(i)\mydef \max_{j\in[b]} p(j) + 1$, update $\A$ and $\R$ accordingly, and then place $x$ into $B_{p(i)}$ using $\onlinesorter{k-1}$.

By the inductive hypothesis, $\onlinesorter{k-1}$ can place $n'$ reals into an array of size $w$.
Hence, to prove correctness we only have to ensure that the set of used boxes are contained in the array, i.e., that $\R \leq (1 + 2k\delta) n$. We show that the total number of empty cells in $[1, \R]$ is at most $2k\delta n$ and therefore there must be $n$ full cells in $[1, (1 + 2k\delta) n]$.

Let $\{B_1, \dots, B_\ell\}$ be the set of boxes we used (either full or active), so that $[1, \R] = \bigcup_{i=1}^\ell B_i$. We partition~$[\ell]$ into $F$ and $E$ so that $F\mydef\ld\{ i \in [\ell] \,|\, B_i \text{ is full}  \rd\}$ and $E \mydef [\ell] \setminus F$.  Denote with $\E (B_i)$ the number of empty cells in box $B_i$. 
If $i\in F$, we have 
$\E(B_i) \leq  w - n' = 2 (k-1) \delta \cdot  n'$. Moreover,
$
|F| \leq n/n'
$
and therefore 
\[
\sum_{i\in F} \E(B_i) \leq \frac{n}{n'} \cdot 2 (k-1) \delta\cdot  n' = 2(k-1)\delta n.
\]
If $i\in E$, we use the trivial bound $\E(B_i) \leq  w$, moreover $j \in E$ implies $B_j \in \A$, hence $|E| \leq b$. This yields
\[
\sum_{i\in E} \E(B_i) \leq wb 
\leq (1+2(k-1)\delta)n'\cdot n^{1/(k+1)}
= (1+ 2(k-1) \delta)\cdot \delta n.\]
Putting everything together, we obtain
\[
\sum_{i\in [\ell]} \E(B_i) \leq 2(k-1) \cdot \delta n + (1 + 2(k-1) \delta)\cdot \delta n \leq 2k \delta n,
\]
where the last inequality holds because $2(k-1)\delta \leq 1$ by assumption on $k$. 

We are left to bound the total cost for which we likewise use induction. By $\cost^{(k)} (r_1, \dots, r_n)$, we denote the cost incurred by algorithm $\onlinesorter{k}$ when facing the stream of reals $r_1, \dots, r_n$. 
We prove that there exists $C > 0$ such that $\cost^{(k)} (r_1, \dots, r_n) \leq \beta \cdot n^{1/(k+1)}  \delta^{-C \cdot (k+1)}$ for any stream of reals  $r_1, \dots, r_n$ with $ r_i\in [\alpha, \alpha + \beta)$.
For $k=1$, we already explained above how \cref{thm:SortingUpperBoundBaby} implies  that $\onlinesorter{1}$ places a stream $r_1, \dots, r_n \in [\alpha, \alpha + \beta)$ with a cost of $\cost^{(1)} (r_1, \dots, r_n) = 18 \beta \sqrt{n}$, and we can choose $C$ accordingly.

Now suppose $k>1$,
and let $\{B_1, \dots, B_\ell\}$ be the set of full or active boxes. We have $\ell \leq \R/w \leq (1+2k\delta) n / w \leq 3 n^{1/(k+1)} \delta^{-1}$. We can think of our algorithm as partitioning the stream $r_1, \dots, r_n$ into 
substreams according to the box in which each real is placed. For $i \in [\ell]$, denote the 
substream of length~$L_i$ placed in~$B_i$ with 
$y^i_1, \dots, y^i_{L_i}$ so that we have $\{r_1, \dots, r_n\} = \bigcup_{i=1}^\ell \{y^i_1, \dots, y^i_{L_i}\}$.
Moreover, $L_i \leq n'$ for each $i \in [\ell]$.
Define $\alpha' \mydef (i-1) \cdot \beta n^{-1/(k+1)}$ and $\beta' \mydef \beta n^{-1/(k+1)}$, then for each $j \in [L_i]$ it holds $\alpha' \leq y^i_j < \alpha' + \beta'$.
By the induction hypothesis, the cost induced by the recursive call of $\onlinesorter{k-1}$ on box $B_i$ is bounded by \vspace{-12pt}

\begin{align*}
\cost^{(k-1)} \ld(y^i_1, \dots, y^i_{L_i}\rd) 
&\leq \beta^\prime 
\cdot \ld(L_i\rd)^{{1}/{k}}
\cdot \delta^{-Ck}\\
&\leq \beta n^{-1/(k+1)} 
\cdot \ld(\delta n^{k/(k+1)}\rd)^{{1}/{k}} \cdot \delta^{-C k} \\
&\leq \beta \cdot \delta^{-C k}.
\end{align*}
Now we are ready to estimate the total cost of $\onlinesorter{k}$ as the sum of costs generated inside a box~$B_i$ or between any two consecutive boxes: 
\vspace{-12pt}

\begin{align*}
\cost^{(k)}\ld(r_1 \dots r_n\rd) 
&\leq \ell\beta + \sum_{i=1}^\ell \cost^{(k-1)} \ld(y^i_1, \dots, y^i_{L_i}\rd) \\
&\leq \ell \beta \cdot \ld(1 + \delta^{-C k} \rd) \\
&\leq 3 n^{1/(k+1)} \delta^{-1} \beta \cdot \ld(1 + \delta^{-C k} \rd) \\
&\leq \beta \cdot n^{1/(k+1)} \delta^{-C \cdot (k+1)}
\end{align*}
where the last inequality holds if we choose $C$ large enough.
\end{proof}

Finally, we can prove \Cref{thm:SortingUpperBound}.
\begin{proof}[Proof of \Cref{thm:SortingUpperBound}]
 We apply \Cref{lem:foranyk} choosing $\alpha \mydef 0$, $\beta \mydef 1$, $k \mydef \ld\lfloor\sqrt{\log n} / \sqrt{\log \log n + \log (1/\varepsilon)} \rd\rfloor \geq 1$ and $\delta \mydef \varepsilon / (2k)$. It uses $(1 + 2k\delta) n \leq (1+\varepsilon) n$ memory cells and yields a cost of 
\[
n^{1/(k+1)} \cdot \delta^{-O(k+1)} = 2^{O\ld({\log n}/{k} + k \log ({2k}/{\varepsilon
})\rd)} = 2^{O\ld(\sqrt{\log n} \cdot \sqrt{\log \log n + \log (1/\varepsilon)}\rd)}.\qedhere
\]
\end{proof}

\section{Online packing}\label{sec:onlinepacking}

In this section, we consider various online packing problems.
In \cref{sec:PackingLower}, we show that \StripPackingConv{} does not allow for a competitive online algorithm.
In fact, the argument generalizes to the other important online packing problems \SquarePacking{}, \PerimeterPacking{} and \BinPacking{}.
This holds even when all pieces have diameter less than an arbitrarily small constant $\delta>0$.

In \cref{sec:PackingAlgo}, we present an online algorithm for convex strip packing. A naive greedy algorithm for this problem places each new piece as deep into the strip as possible, and this algorithm is $n$-competitive, where $n$ is the number of pieces.
Our algorithm is $O(n^{0.59})$-competitive.

\subsection{Lower bounds --- no competitive algorithms}\label{sec:PackingLower}

\cref{thm:TheLowerBound} and \cref{lemma:stripToArrayComplete} yield the following corollary.

\begin{corollary}\label{cor:StripPacking}
\StripPackingConv{} does not allow for an asymptotically competitive online algorithm, even when the diameter of each piece is at most $2$.
Specifically, for every online algorithm $\A$, there exists a stream of $n$ parallelograms of diameter at most $2$ such that $\A$ produces a packing of width $\Omega(\sqrt{\log n/\log\log n})$, while the  optimal offline packing has width at most $2$. 
\end{corollary}

\begin{proof}
Suppose that \StripPackingConv{} has a $C$-competitive algorithm $\A$.
By \cref{lemma:stripToArrayComplete}, we get a $4C$-competitive algorithm $\A_2$ for \OnlineSorting{}$[2C,n]$.
\cref{thm:TheLowerBound} implies that $8C^2=\Omega(\log n /\log \log n)$.
In particular, $C=\Omega(\sqrt{\log n /\log \log n})$. 

Specifically, the proof of \cref{thm:TheLowerBound} yields a stream of $n$ reals where $\A_2$ incurs a cost of $\Omega(\sqrt{\log n /\log \log n})$.
By the proof of \cref{lemma:stripToArrayComplete}, this translates to a stream of $n$ parallelograms that can be packed in a strip of width $2$, while $\A_1$ produces a packing of width $\Omega(\sqrt{\log n /\log \log n})$.
\end{proof}

The insights of 
\cref{cor:StripPacking}
imply negative answers also for various other online packing problems.

\PackingMasterThm*
\begin{proof}
For each of the four packing problems, we consider an arbitrary algorithm $\A^*$, where $\A^*$ packs pieces into a container $S^*$.
Here, $S^*$ may be a strip, a set of bins etc., depending on the problem at hand.
The main idea in all the proofs is to \emph{cover} $S^*$ by (rotated) substrips of a strip $S$ of height $1/c$ for a (large) constant $c$, so that we get a correspondence between points in $S^*$ and $S$.
By feeding $\A^*$ with a stream of parallelograms and observing where they are placed in $S^*$, we get an online algorithm $\A$ for packing (slightly modified) parallelograms into $S$.
We can then by \cref{cor:StripPacking} choose a stream that forces $\A$ to produce a packing much larger than the optimal one, which implies that $\A^*$ has likewise produced a bad packing.
This gives a lower bound on the competitive ratio of $\A^*$.

For a horizontal parallelogram $P$, a \emph{2-extended} copy is a parallelogram obtained by taking two copies of $P$ and identifying the bottom base segment of one with the top segment of the other copy.
We now prove each statement in the theorem.
\begin{enumerate}[(a)]
\item
Suppose that $\A^*$ is an
algorithm for \StripPackingConv\ restricted to pieces of diameter at most $\delta>0$.
The algorithm $\A^*$ packs pieces into a strip $S^*$ of height $1$.
We define $c\mydef \lceil 4/\delta\rceil$.
Recall that $S$ is a strip of height $1/c$.
We cut $S$ into substrips $S_1,S_2,\ldots$, each of which is a rectangle of size $1\times 1/c$.
We rotate these by $90^\circ$ and use them to cover the strip $S^*$.
Let $S^*_i$ be the part of $S^*$ corresponding to $S_i$, so that the rectangles $S^*_1,S^*_2,\ldots$ appear in this order from left to right in $S^*$; see \cref{fig:Covering}.

\begin{figure}[htb]
\centering
\includegraphics[page=5]{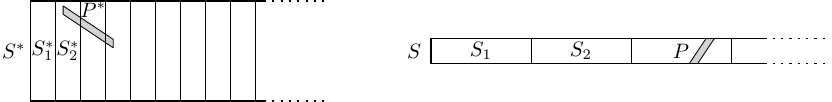}
\caption{The figure shows the correspondence between the strips $S^*$ and $S$.}
\label{fig:Covering}
\end{figure}

We now define an algorithm $\A$ for strip packing horizontal parallelograms of height $1/c$ and diameter at most $2/c$ into $S$ as follows.
Let $P$ be a piece to be packed in $S$, and let $P^*$ be the piece obtained from rotating the $2$-extended copy of $P$ by $90^\circ$.
Then $P^*$ has diameter at most $4/c\leq \delta$.
We now feed $P^*$ to~$\A^*$.
We observe where $\A^*$ places $P^*$ in $S^*$.
Since $P^*$ has width $2/c$, it intersects both the left and right vertical edge of a substrip $S^*_i$.
Then, in particular, $P'\mydef P^*\cap S^*_i$ is congruent to $P$.
We now place $P$ in the substrip $S_i$ as specified by the placement of $P'$ in $S^*_i$.
As $\A^*$ does not place pieces in $S^*$ so that they overlap, this approach will produce a valid packing in $S$.

By \cref{cor:StripPacking}, there exists for arbitrarily large values $\beta>0$ and $n\in \mathbb N$ a stream $I$ of $n$ pieces that can be packed within a rectangle in $S$ of size $2/c\times 1/c$, whereas $\A(I)=\Omega(\sqrt{\log n/\log \log n})$.
Let $I^*$ be the stream of $2$-extended pieces which we feed to $\A^*$.
To get an upper bound for $\OPT(I^*)$, we consider the aforementioned packing of the pieces $I$ in a rectangle of size $2/c\times 1/c$.
Enlarging the pieces, we get a packing of the 2-extensions in a rectangle of size $3/c\times 2/c$.
Rotating the rectangle, we get a packing of the pieces $I^*$ in a rectangle of size $2/c\times 3/c\subset 1\times 3/c$, assuming $c\geq 2$.
Hence, $\OPT(I^*)\leq 3/c$.
Since $\A(I)\leq c\cdot \A^*(I^*)$, we also have $\A^*(I^*)=\Omega(\sqrt{\log n/\log \log n})$.
The statement then follows.

\item
The proof is almost identical to that of (a), so we only describe the parts that are different.
Suppose that $\A^*$ is an
algorithm for \BinPacking\ restricted to pieces of diameter at most $\delta>0$.
Here, $\A^*$ packs pieces into unit square bins $B_1,B_2,\ldots$.

\begin{figure}[htb]
\centering
\includegraphics[page=16]{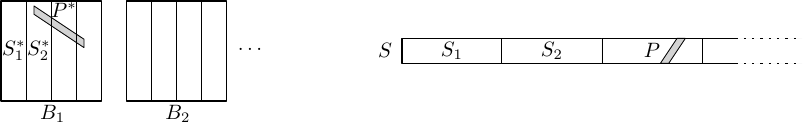}
\caption{The figure shows the correspondence between the bins $B_1,B_2,\ldots$ and $S$.}
\label{fig:Covering2}
\end{figure}

We again partition $S$ into substrips and cover the boxes $B_1,B_2,\ldots$ with these, as shown in \cref{fig:Covering2}.
We then get a strip packing algorithm $\A$ in $S$ in a similar way as in (a), and obtain
a lower bound on the asymptotic competitive ratio of $\A^*$ in a similar way.

\item
Suppose there exists a $C$-competitive online algorithm $\A^*$ for \PerimeterPacking.
We prove that this implies the existence of an algorithm for \SquarePacking{}$[1/160C^2]$, and it follows from (d) that $C=\Omega(\sqrt[4]{\log n/\log\log n})$.
By \cref{thm:offline} (c), every set of convex polygons of diameter at most $1/10$ and total area at most $1/10$ can be packed in a $1\times 1$ square.
Scaling by $1/4C$, we get that convex polygons of diameter at most $1/40C$ and total area at most $1/160C^2$ can be packed in a $1/4C\times 1/4C$-square.
In particular, the same holds if we bound the diameters to be at most $1/160C^2$.
Hence, such a set of polygons can be packed in a box of perimeter $1/C$.
It follows that $\A^*$ will produce a packing with a bounding box of perimeter at most $1$.
Therefore, the packing will be contained in the unit square centered at any corner of the first piece $P_1$.
In other words, we have defined an algorithm for the problem \SquarePacking{}$[1/160C^2]$, and the claim follows.

\item
Consider an online algorithm $\A^*$ for packing a stream of pieces of diameter at most $\delta\in (0,1]$ into the unit square.
For arbitrarily large values of $n$, we prove that there exists a stream of pieces of diameter at most $\delta$ and total area $O(\sqrt{\log\log n/\log n})$ that $\A^*$ cannot pack into the unit square.
Let $d>0$ be a lower bound on the multiplicative constant hidden in the $\Omega$-symbol of Corollary~\ref{cor:StripPacking}.
We use the same setup as in (b), just with a single box $B$, which is covered by the $c\mydef \lfloor \sqrt[4]{d^2\log n/\log\log n}\rfloor$ first substrips of $S$.
By Corollary~\ref{cor:StripPacking}, there exists a stream $P_1,\ldots,P_n$ of pieces such that $\A$ produces a packing of width $\sqrt[4]{d^2\log n/\log\log n}$, so that $\A^*$ cannot fit the stream $I^*$ of $2$-extended pieces in $B$.
We may without loss of generality assume that $n$ is sufficiently large that $4/c\leq \delta$, so that the $2$-extended pieces have diameter at most $\delta$.
Each piece in $I^*$ has area $1/c^2n$, so the pieces have total area $1/c^2=O(\sqrt{\log\log n/\log n})$, and the statement follows.\qedhere
\end{enumerate}
\end{proof}

\subsection{A better-than-naive algorithm for strip packing}\label{sec:PackingAlgo}

You are asked to suggest an algorithm for online translational strip packing of convex pieces.
What is the first approach that comes to your mind?
It may very well be the following greedy algorithm:
For each piece $P_i$ that appears, place $P_i$ as far left into the (horizontal) strip as possible.
This algorithm is $n$-competitive:
Indeed, it will occupy no more than $n\cdot \max_i \pwidth (P_i)$ of the strip, and the optimum must occupy at least $\max_i \pwidth (P_i)$.
This bound is unfortunately also essentially tight:
Consider a sequence of very skinny pieces of height $1$ and width $1$, but with slopes alternating between $1$ and $-1$.
The algorithm will produce a packing of width $n$, while the optimum has width slightly more than $2$ (depending on the actual fatness of the pieces).

We found it surprisingly difficult to develop an algorithm provably better than the naive algorithm, but in the following, we will describe an $O(n^{\log 3-1}\log n)$-competitive algorithm (note that $\log 3-1<0.59$).
We denote the algorithm \onlinepacker.

We first describe the algorithm and carry out the analysis when all the pieces are parallelograms of a restricted type and then show, in multiple steps, how the method generalizes to arbitrary convex pieces.
By rescaling, we may without loss of generality assume that the first piece presented to the algorithm has width~$1$.

\subsubsection{Algorithm for horizontal parallelograms of height $1$ and width $\leq 1$}

In the following we describe and analyze the algorithm \onlinepacker\ for online translational strip packing under the assumption that all pieces $P$ are horizontal parallelograms where $\pheight (P)=1$ and $\pwidth (P)\leq 1$.

\begin{figure}[htb]
\centering
\includegraphics[page=6]{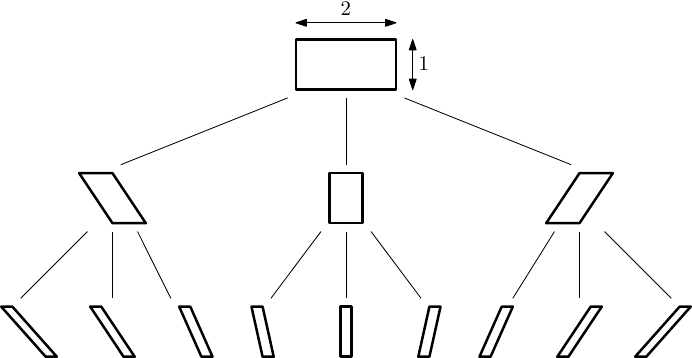}
\caption{The top three layers of the box type tree.}
\label{fig:boxtypes}
\end{figure}

We define an infinite family of \emph{box types} as follows.
The box types can be represented by an infinite ternary \emph{box type tree}; see \cref{fig:boxtypes}.
The empty vector $[\;]$ denotes the \emph{basic box type}, which is simply a rectangle of size $2\times 1$.
Each box type will be defined as a horizontal parallelogram of height~$1$.
A box type is represented by a vector $[x_1,\ldots,x_d]\in\{-1,0,+1\}^d$.
Given a $d$-dimensional box type $T\mydef[x_1,\ldots,x_d]\in\{-1,0,+1\}^d$, we define $(d+1)$-dimensional box types
$T\oplus[x_{d+1}]=[x_1,\ldots,x_d,x_{d+1}]$ for $x_{d+1}\in\{-1,0,+1\}$, as follows.
Let $b$ and $t$ be the bottom and top edges of $T$.
We partition $b$ and $t$ into three equally long segments $b_{-1},b_0,b_{+1}$ and $t_{-1},t_0,t_{+1}$, respectively, in order from left to right.
For $x_{d+1}\in\{-1,0,+1\}$, we then define $T\oplus[x_{d+1}]$ to be the parallelogram spanned by the segments $b_0$ and $t_{x_{d+1}}$.
It follows that a box of type $[x_1,\ldots,x_d]$ has base edges of length $2\cdot 3^{-d}$ and the other pair of parallel edges have width $2\sum_{i=1}^d \nicefrac{x_i}{3^i}$. 

\begin{lemma}\label{lemma:containSegment}
For every $d\geq 0$ and every line segment $s$ of height $1$ and width at most $1$, there exists a $d$-dimensional box type $T$ that can contain $s$ when the lower endpoint of $s$ is placed at the midpoint of the bottom segment of $T$.
\end{lemma}

\begin{proof}
We prove the lemma by induction on $d$.
The claim clearly holds for $d=0$, because the basic box type is a rectangle of size $2\times 1$ and $s$ has height $1$ and width at most $1$.
Suppose that for some dimension $d\geq 0$, there is a box type $T$ that satisfies the lemma.
We then partition the bottom and top edges of $T$, as described in the definition of the box types.
We choose $x_{d+1}\in\{-1,0,+1\}$ such that $t_{x_{d+1}}$ contains the upper endpoint of $s$. Then $T\oplus[x_{d+1}]$ is a $(d+1)$-dimensional box type that satisfies the claim.
\end{proof}

\begin{lemma}\label{lemma:suitablebox}
For every horizontal parallelogram $P$ where $\pheight (P)=1$ and $\pwidth (P)\leq 1$, there exists a box type $T$ such that $P$ can be packed into $T$ and $\area (T) \leq 6\area (P)$.
\end{lemma}

\begin{proof}
Let $\ell$ be the length of the horizontal segments of $P$.
Choose $d\geq 0$ as large as possible such that $3^{-d}\geq \ell$; this is possible because $\ell\leq 1$.
Let $s$ be a segment of height $1$ parallel to the non-horizontal edges of $P$ and consider the $d$-dimensional box type that contains $s$ as described in \cref{lemma:containSegment}.
If the top endpoint of $s$ is in the right half of the top edge of $T$, then we place $P$ to the left of $s$, i.e., with the right edge of $P$ coincident with $s$.
Otherwise, we place $P$ to the right of $s$.
Since the horizontal segments of $T$ have length $2\cdot 3^{-d}\geq 2\ell$, it follows that $T$ can contain $P$.
Moreover, by maximality of $d$, the base edges have length less than $6\ell$, so it follows that $\area (T) \leq 6\area (P)$.
\end{proof}

If a piece $P$ can be packed into a box type $T$ and $\area (T)\leq 6\area (P)$, as in the lemma, then we say that $T$ \emph{matches} $P$.
We say that a box type $T$ is \emph{suitable} for a piece $P$ if $T$ is an ancestor in the box type tree of a box type that matches $P$.
In particular, $P$ can be packed in a box $B$ of a suitable type, but the area of $B$ may be much more than $6\area(P)$.
We consider a type $T$ to be an ancestor of itself.

Our algorithm will allocate space of the strip for boxes of the various box types.
Each box of type $T$ contains either:
\begin{itemize}
\item a piece $P$ that $T$ matches, or

\item
one, two, or three boxes of types that are children of $T$ in the box type tree.
\end{itemize}

\begin{figure}[htb]
\centering
\includegraphics[page=7]{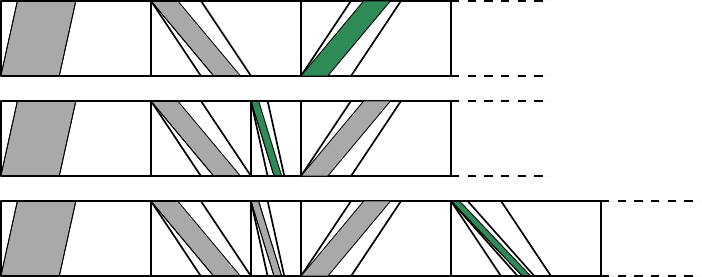}
\caption{An evolving packing produced by the algorithm.
In each case, the green piece has just been placed.
In the third step, there was no box suitable for the piece which also had room for it, so we allocated a new basic box.}
\label{fig:algpack}
\end{figure}

We now explain the behavior of our algorithm \onlinepacker\ when a new piece $P$ arrives; see \cref{fig:algpack}.
Let $T=[x_1,\ldots,x_d]$ be a type that matches $P$.
Consider first the special case that $d=0$, which means that the basic box type matches $P$.
We then allocate a new box $B_0$ of the basic type and place it as far left in the strip as possible (that is, without overlapping any already allocated box) and place $P$ there.

Otherwise, for $i\in\{0,\ldots,d\}$, let $T_i\mydef [x_1,\ldots,x_i]$, so that $T_0,\ldots,T_d$ is the path in the box type tree from the basic box type $T_0 \mydef [\; ]$ to $T_d\mydef T$.
Suppose first that for some $i<d$, there exists a box $B_i$ of type $T_i$ in the packing that has room for one more box of type $T_{i+1}$.
Choose $i$ to be maximum with this property, which means that $B_i$ is chosen as small as possible.
We then do the following for each $j=i+1,\ldots,d$:
Allocate a box $B_j$ of type $T_j$ in $B_{j-1}$, which is placed as far left as possible.
At last, we place $P$ in $B_d$.
Thus, each of the new boxes $B_{i+1},\ldots,B_{d-1}$ will contain a single box, and $B_d$ will contain $P$.

If such a dimension $i$ does not exist, we allocate a new box $B_0$ of the basic type and place it as far left in the strip as possible.
It then holds for $i=0$ that $B_i$ has type $T_i$ and has room for a box of type $T_{i+1}$, so we can proceed as described above; allocating a chain of nested boxes until we get to a box of type $T_d$ and place $P$ there.

We say that a $d$-dimensional box is \emph{near-empty} if there is exactly one $(d+1)$-dimensional box allocated in it.
A crucial property of \onlinepacker\ is that it does not produce an excessive number of near-empty boxes, as described in the following lemma.

\begin{figure}[htb]
\centering
\includegraphics[page=8]{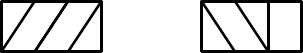}
\caption{For any type $T$ it holds that a box of type $T$ can contain two boxes $B_1,B_2$ of types that are children of $T$ if $B_1$ and $B_2$ have the same type or one of them has type $T\oplus [0]$.
Here, it is shown for the base type $T\mydef[\; ]$.}
\label{fig:twoChildren}
\end{figure}

\begin{lemma}\label{lemma:atmosttwo}
There can be at most two near-empty boxes of each type in a packing produced by \onlinepacker.
\end{lemma}

\begin{proof}
Suppose that there are two near-empty boxes $B_1$ and $B_2$ of type $T$, where $B_1$ was created first.
Let $P$ be the piece that caused $B_2$ to be allocated.
We first see that one of the two boxes, say $B_1$, must contain a box $B_1'$ of type $T\oplus[-1]$ and the other, $B_2$, must contain a box $B_2'$ of type $T\oplus[+1]$:
If $B'_1$ and $B'_2$ had the same type or one of them had type $T\oplus[0]$, then they could be packed in the same box of type $T$; see \cref{fig:twoChildren}.
Hence, $P$ could have been placed in $B_1$ instead of allocating the new box $B_2$, contradicting that we allocate new boxes in a smallest possible existing box.

But since $B'_1$ has type $T\oplus[-1]$ and $B'_2$ has type $T\oplus[+1]$, it follows that the next box of type $T\oplus[x]$, for $x\in\{-1,0,+1\}$, can be placed in $B_1$ or $B_2$.
Therefore, the algorithm will never allocate a third near-empty box of type $T$.
\end{proof}

We will now analyze the density under the assumption that there are no near-empty boxes.
We shall then reduce the general case to this restricted case.

\begin{lemma}\label{lemma:no-near-empty}
Suppose that \onlinepacker\ has produced a packing of $n$ horizontal parallelograms of height $1$ and width at most $1$, where there are no near-empty boxes.
Then the density of the pieces in the occupied part of the strip is at least $\Omega(n^{1-\log 3})$.
\end{lemma}

\begin{proof}
We prove that when there are no near-empty boxes, then the density in each basic box is at least $n^{1-\log 3}$, and the lemma follows.
So consider a basic box $B$ and let $n_B$ be the number of pieces packed in $B$.
The box $B$ and the boxes allocated in it can be represented as a rooted tree $\mathcal T_1$. 
Here, the root is $B$ and the children of a $d$-dimensional box are the $(d+1)$-dimensional boxes that it contains, and each leaf is a box that contains a piece.
Since there are no near-empty boxes, each internal node in $\mathcal T_1$ has two or three children.
To get a lower bound on the density in $B$, we construct a sequence of trees $\mathcal T_1,\mathcal T_2,\ldots,\mathcal T_k$, where each tree corresponds to a packing in a basic box.
The density of these packings decreases, and we will see in the end that the density in $\mathcal T_k$ is at least $\Omega(n^{1-\log 3})$.

The final tree $\mathcal T_k$ is balanced in the sense that all leafs are contained in two neighbouring levels $d$ and $d+1$.
If a tree $\mathcal T_i$ is not balanced, we construct $\mathcal T_{i+1}$ in the following way.
Let $d$ be the maximum dimension of a box in $\mathcal T_i$.
We can find two or three leafs of dimension $d_u$ that have the same parent~$u$. 
If there are three, we remove one of them and the unique piece it contains from the packing, which decreases the density, and then proceed as in the case where there are two, as described in the following.
Let $u_1,u_2$ be the two leafs.
Suppose that there is also a leaf $v$ at level $d_v\leq d_u-2$. 
Recall that a $d$-dimensional box has area $ 2\cdot 3^{-d}$. Thus, the area of pieces in the boxes $u_1,u_2,v$ combined is at least
\[
A\mydef 2\cdot 2\cdot 3^{-d_u}/6+2\cdot 3^{-d_v}/6=2\cdot 3^{-(d_u+1)}+ 3^{-(d_v+1)}.
\]

We now make a replacement argument:
We ``move'' the leafs $u_1,u_2$, so that they become children of $v$ instead of $u$, and argue that this only decreases the density.
So, we construct a tree $T_{i+1}$ that is similar to $\mathcal T_i$, except $u$ is now a leaf and $v$ has two children $v_1,v_2$. For this to be a conceivable packing, we must remove the pieces that were before in $u_1,u_2,v$ and place new pieces in $u, v_1, v_2$.
We place pieces in $u, v_1, v_2$ that fills the boxes with density $1/6$.
These new pieces have total area
\[
A'\mydef 2\cdot 3^{-d_u+1}/6+2\cdot 2\cdot 3^{-d_v-1}/6=3^{-d_u}+ 2\cdot 3^{-(d_v +2)}.
\]
It is now straightforward to verify that since $d_v+2\leq d_u$, we have $A'\leq A$. 
Hence, the density in $\mathcal T_{i+1}$ is smaller than the density in $\mathcal T_i$.

In the end, we obtain the packing represented by a balanced tree $\mathcal T_k$, where all leafs are $d$- or $(d+1)$-dimensional for some value $d\geq 0$.
We then have $d\leq \log n_B+1$.
Therefore, each piece has area at least $2\cdot 3^{-\log n_B-1}/6=\Omega(3^{-\log n})$, and the total area of the pieces is $\Omega(3^{-\log n_B}\cdot n_B)=\Omega(n_B^{1-\log 3})=\Omega(n^{1-\log 3})$, which is also a lower bound on the density.
\end{proof}

We now prove that when the total area of pieces is at least $1$, then the density of an arbitrary packing produced by the algorithm, i.e., where there may also be near-empty boxes, is not much smaller than stated in \cref{lemma:no-near-empty}.
If the total area is not bounded from below, it is easy to see that no bound can be given on the density; a single skew and sufficiently thin parallelogram proves that.
We don't know, however, if the algorithm is anyway $O(n^{\log 3-1}\log n)$-competitive.

\begin{lemma}\label{lemma:oneclass}
Suppose that \onlinepacker\ has packed $n$ horizontal parallelograms of height $1$ and width at most $1$.
If the total area of the pieces is at least $1$, the resulting packing has density $\Omega(n^{1-\log 3}/\log n)$.
\end{lemma}

\begin{proof}
Let $A$ be the area of the strip used by the algorithm, and let $S$ be the total area of the pieces.
We show that by replacing some boxes and pieces, so that the total area of pieces increases by at most $F=O(\log n)$, we obtain a packing with $m$ pieces, where $m\leq n$, and no near-empty boxes.
We then get from \cref{lemma:no-near-empty} that the resulting packing has density
\[
\frac{S+F}{A}=\Omega(m^{1-\log 3})=\Omega(n^{1-\log 3}).
\]
Using that $S\geq 1$, it then follows that the original density is 
\[
\frac{S}{A}=\frac{S+F}{A}\cdot \frac{S}{S+F}
\geq \frac{S+F}{A}\cdot \frac{S}{S+ O(\log n)}
\geq \frac{S+F}{A}\cdot \frac{S}{S\cdot O(\log n)}=
\Omega(n^{1-\log 3}/\log n).
\]

Consider a maximal near-empty box $B_1$, i.e., a near-empty box that is not contained in a larger near-empty box.
We remove all boxes contained in $B_1$ and the pieces they contain and instead place a single piece in $B_1$ that completely fills $B_1$.
This operation cannot increase the number of pieces, because there was at least one piece contained in $B_1$ before.

Let $d$ be maximum such that there are at least $2^d$ near-empty $d$-dimensional boxes.
We first observe that $d\leq \log n$:
Each of the near-empty $d$-dimensional boxes contains a distinct piece.
We therefore have $2^d\leq n$, so that $d\leq \log n$.

Let $F_{\leq d}$ be the total area of pieces that we add to eliminate maximal near-empty boxes of dimension at most $d$ and let $F_{> d}$ be the remaining ones, so that $F=F_{\leq d}+F_{>d}$.

In order to bound $F_{\leq d}$, we note that there are $3^i$ types of $i$-dimensional boxes, each of which has area $O(3^{-i})$.
For each maximal near-empty $i$-dimensional box, we add a piece of area $O(3^{-i})$.
As there are at most two near-empty boxes of each type by Lemma~\ref{lemma:atmosttwo}, we add boxes of total area $O(3^i\cdot 3^{-i})=O(1)$ for each $i\leq d$.
As $d=O(\log n)$, we have $F_{\leq d}=O(\log n)$.

By similar arguments, we get
\[
F_{>d}<\sum_{i=d+1}^\infty 2^i\cdot O(3^{-i})=\sum_{i=d+1}^\infty O((2/3)^i)=O(1),
\]
so we conclude that $F=O(\log n)$, which finishes the proof.
\end{proof}

\subsubsection{Algorithm for horizontal parallelograms of extended width $\leq 1$}\label{sec:smallWidth}

We now describe an extension of \onlinepacker\ in order to handle horizontal parallelograms of arbitrary height (at most $1$) and bounded extended width, to be defined shortly.
Here, we partition the pieces into height classes, so that a piece $P$ belongs to class $h$ if $2^{-h-1}< \pheight (P)\leq 2^{-h}$.
For height class $h$, the base type is a rectangle of size $2\times 2^{-h}$, and in the strip we allocate boxes of this size in which to pack pieces from the height class.
We define an infinite ternary box type tree for each height class, exactly as when all the pieces have height $1$.

When a piece $P$ arrives, we determine the height class $h$ of $P$.
We extend the non-horizontal segments of $P$ until we obtain a horizontal parallelogram of height exactly $2^{-h}$, and we denote the resulting piece $P'$ so that $P\subset P'$.
We then define the \emph{extended width} of $P$ as $\pextwidth(P)\mydef \pwidth(P')$.
Note that since $2^{-h-1}< \pheight (P)$, we have $\pheight (P')<2\pheight (P)$ and $\pextwidth(P)=\pwidth(P')<2\pwidth(P)$.
We then pack $P'$ (including $P$) into a minimum suitable box that has already been allocated, if possible, and otherwise allocate a new basic box.

We stack the basic boxes of different height classes if possible.
When a new basic box of height class $h$ is created, we stack it on the leftmost pile of basic boxes that has room for it.
This choice implies that there can be at most one stack of basic boxes that is less than half full. 

\begin{lemma}\label{lemma:boundedWidth}
Suppose that \onlinepacker\ 
has packed $n$ horizontal parallelograms of width at most $1/2$ and total area more than $2$.
The resulting packing has density $\Omega(n^{1-\log 3}/\log n)$.
\end{lemma}

\begin{proof}
Let $U$ be the union of the base boxes (across all height classes), and let $A$ be the area of the part of the strip occupied by the packing.
Since the total area of pieces is more than $2$, we also have $\area(U)\geq2$.
Suppose first it holds that the density in $U$ is at least $\Omega(n^{1-\log 3}/\log n)$.
We have that $\area(A)\leq 2\area(U)+2$, since there is at most one stack of base boxes of height less than $1/2$ in the packing.
Since $2<U$, we then also have $\area(A)<3\area(U)$.
We can therefore also conclude that the density of the occupied part of the strip is $\Omega(n^{1-\log 3}/\log n)$.

We now prove the density bound in $U$.
Let $S$ be the total area of the pieces.
For each height class $h$ for which there are some pieces, we feed the algorithm with a rectangle of size $1\times 2^{-h}$.
Let $F$ be the total area of these pieces, and we have $F<\sum_{h=0}^\infty 2^{-h} = 2$.
We therefore have $2S\geq S+F$.
Let $U'$ be the union of the base boxes in this expanded packing.
We then have that the original density is
\[
\frac{S}{\area(U)}\geq \frac{S+F}{2\area(U)}\geq \frac{S+F}{2\area(U')}.
\]

Since the area of pieces in each non-empty height class $h$ is at least $2^{-h}$, we can now apply Lemma~\ref{lemma:oneclass} (by scaling the $y$-coordinates by a power of $2$, we obtain a packing of parallelograms of height $1$).
Let $n_h$ be the number of pieces in height class $h$.
We get that the density in the base boxes of height $2^{-h}$ is $\Omega(n_h^{1-\log 3}/\log n_h)=\Omega(n^{1-\log 3}/\log n)$, so this is a bound on the density in all of $U'$.
Hence we have $\frac{S}{\area U}=\Omega(n^{1-\log 3}/\log n)$, and this concludes the proof.
\end{proof}

\subsubsection{Algorithm for all horizontal parallelograms}\label{sec:allHorizontal}

We now describe an extension of \onlinepacker\ that handles horizontal parallelograms of arbitrary height and arbitrary width.
We partition the parallelograms into \emph{width classes}.
Pieces of width class $i$ are packed in base boxes of width $2^i$.
Consider a parallelogram $P$.
If $\pextwidth(P)\leq 1$, then the width class of $P$ is $i=1$.
Otherwise, $P$ belongs to the width class $i$ such that $2^{i-1} < 2\pextwidth(P) \leq 2^i$.
We handle each width class independently, so that pieces from each class are packed in stacks as described in Section~\ref{sec:smallWidth}.
In particular, each width class $i$ is subdivided into height classes, and we allocate in the strip rectangles of size $2^i\times 1$, where we can place a stack of base boxes of sizes $2^i\times 2^{-h}$, for various values of $h\geq 0$. 

When a new piece $P$ arrives, we determine its width class $w$ and height class $h$. Then, we pack it into a rectangle of size $2^w\times 1$, in a base box of size $2^w\times 2^{-h}$ that has room for it. If no rectangle of that size has room for $P$, a new rectangle of size $2^w \times 1$ is created and placed as far left in the strip as possible.

\begin{lemma} \label{lem:better_than_naive_all_classes}
The competitive ratio of \onlinepacker\ when applied to $n$ horizontal parallelograms of arbitrary width is $O(n^{\log 3-1}\log n)$.
\end{lemma}

\begin{proof}
After \onlinepacker\ has packed the $n$ parallelograms, we denote the cost of the produced solution as $\ALG$ and the optimal offline solution as $\OPT$.
We then feed the algorithm with four rectangles of size $2^{i-2}\times 1$ for each width class $i$ for which there are some pieces.
We denote the cost of the resulting packing as $\ALG^+$ and the cost of the optimal offline solution as $\OPT^+$.
Suppose that the largest width class is class $k$.
We now claim that some piece $P$ has extended width more than $2^{k-2}$.
If $k>1$ this holds for any piece in width class $k$, and if $k=1$, it follows since we assume that the first piece presented to the algorithm has width $1$.
We then have $\OPT\geq\pwidth(P)>\pextwidth(P)/2> 2^{k-3}$.
Since the added pieces are rectangles of height $1$, the optimal packing is similar to the optimal packing for the original instance with the extra pieces added in the end.
We therefore have
\[
\OPT^+\leq \OPT+\sum_{i=1}^k 2^i<\OPT+2^{k+1}< 17\cdot \OPT.
\] 

Let $n_i$ and $S_i$ be the number and total area of the pieces in width class $i$, respectively.
Note that $S_i>2^i$, so that we can apply Lemma~\ref{lemma:boundedWidth} to each width class (under proper scaling).
The bound on the competitive ratio becomes
\[
\frac{\ALG}{\OPT} \leq \frac{\ALG^+}{\OPT^+/17}
\leq \frac{17\sum_{i=1}^k n_i^{\log 3-1}\log n_i\cdot S_i}{\OPT^+}\leq
\frac{17n^{\log 3-1}\log n\cdot S}{S}
=O(n^{\log 3-1}\log n).
\]
\end{proof}

\subsubsection{Algorithm for all convex pieces}

\begin{figure}
\centering
\includegraphics[page=9]{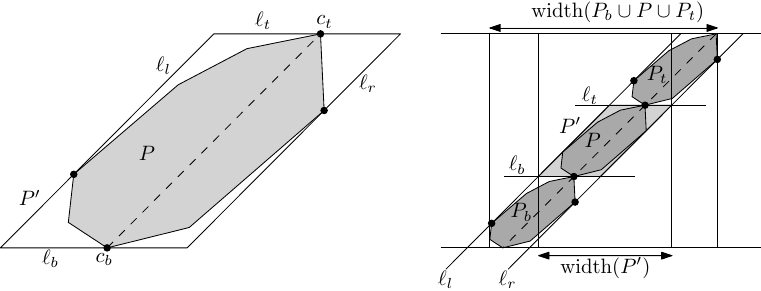}
\caption{Left: The horizontal parallelogram $P'$ has area at most twice that of the piece $P$.
Right: We have $\pwidth(P')\leq \pwidth(P_b\cup P\cup P_t)\leq 3\pwidth(P)$.}
\label{fig:makeParallelogram}
\end{figure}

We now describe the extension of \onlinepacker\ to handle arbitrary convex pieces.
When a new piece $P$ arrives, we find a horizontal parallelogram $P'$ such that $P\subset P'$, $\area (P')\leq 2\area(P)$ and $\pwidth(P') \leq 2 \pwidth(P)$; then we apply \onlinepacker\ to the parallelogram $P'$ (with $P$ inside).

We define $P'$ as follows; see \cref{fig:makeParallelogram} (left). Let $\ell_b$ and $\ell_t$ be the lower and upper horizontal tangent to $P$, respectively, and let $c_b\in P\cap \ell_b$ and $c_t\in P\cap \ell_t$.
Let $\ell_l$ and $\ell_r$ be the left and right tangent to $P$ parallel to the segment $c_bc_t$, respectively.
We then define $P'$ to be the horizontal parallelogram enclosed by the lines $\ell_b,\ell_t,\ell_l,\ell_r$.

It is straightforward to check that $\area (P')\leq 2\area (P)$, because $P$ is convex.
Now we prove that $\pwidth(P')\leq 3\pwidth(P)$; see \cref{fig:makeParallelogram} (right).
Define the translations $P_t\mydef P+(c_t-c_b)$ and $P_b\mydef P-(c_t-c_b)$.
Since $P$ shares points $c_t$ and $c_b$ with $P_t$ and $P_b$, respectively, we have $\pwidth(P_b\cup P\cup P_t)\leq 3 \pwidth(P)$.
Since $P_t$ contain points above $\ell_t$ on both $\ell_l$ and $\ell_r$ and $P_b$ contain points below $\ell_b$ on both $\ell_l$ and $\ell_r$, we have $\pwidth(P_b\cup P\cup P_t)\geq\pwidth(P')$, and we conclude $\pwidth(P')\leq 3 \pwidth(P)$.

In the proof of \Cref{lem:better_than_naive_all_classes} we only bound $\OPT$ using the width of the widest piece and the total area $S$.
With respect to both width and area, the value for an arbitrary convex pieces $P$ is at least $1/2$ of that of the containing parallelogram $P'$.
Therefore, the bound on the competitive ratio carries over up to constant factors when the algorithm is applied to the parallelograms $P'$, as stated in the following.

\PackingUpperBound*

\section{Constant-factor approximations for offline packing}\label{sec:offlinepacking}

In this section we show how to obtain offline approximation algorithms for the packing problems of \cref{thm:PackingMasterThm}.

\offline*

\begin{proof}
Alt, de Berg, and Knauer~\cite{alt_convexOffline_JoCG,alt_convexOffline_corr} presented an algorithm that packs any set $\mathcal P$ of convex polygons (with $n$ vertices in total, using $O(n \log n)$ time) into an axis-aligned rectangular container~$B$ such that area$(B)\leq 23.78\cdot \textsc{opt}$, where \textsc{opt} is the minimum area of any axis-aligned rectangular container for $\mathcal P$.
As an intermediate step, they obtain a collection of rectangular mini-containers that together contain all objects of $\mathcal P$.
Let $w_{\max}$ and $h_{\max}$ denote the maximum width and height among all objects of~$\mathcal P$, respectively.
For some fixed constants $c>0,\alpha\in(0,1)$, each mini-container has width $(c+1)w_{\max}$ and a height $h_i$ where $h_i\mydef \alpha^i h_{\max}$ for some appropriate $i$.
The total area $A_C$ of all mini-containers  can be bounded by 
\begin{equation}\label{eq:offlineeq}
A_C\leq
(1+\nicefrac{1}{c})\cdot \left[\frac{2}{\alpha}\cdot \area(\mathcal P)+ \frac{c+\nicefrac{2}{\alpha}}{1-\alpha} \cdot h_{\max} w_{\max}\right].
\end{equation}
In order to prove (a), (b), (c), and (d), we repeatedly make use of the mini-containers and this inequality.\medskip

We first consider strip packing and prove (a).
When packing a set of convex polygons $\mathcal P$ into a horizontal strip of height 1, the minimal width \textsc{opt-w} is at least $\max\{w_{\max}, \area(\mathcal P)\}$.
Therefore,

\[
A_C\leq
(1+\nicefrac{1}{c})\cdot \left[\frac{2}{\alpha}+ \frac{c+\nicefrac{2}{\alpha}}{1-\alpha}\right]\cdot\textsc{opt-w}.
\]

We group the mini-containers greedily into stacks of height at most 1, which we place in the strip.
Note that all but possibly one stack have a height of  at least $\nicefrac{1}{2}$; otherwise two of these should have been put together. 
Therefore, the number of stacks is at most $\frac{2A_C}{(c+1)w_{\max}}+1$.
Together with $\alpha\mydef 0.49$ and $c\mydef 2.2$, this translates into a width of 

\[
2A_C+(c+1)w_{\max}\leq
\left(2\cdot (1+\nicefrac{1}{c})\cdot \left[\frac{2}{\alpha}+ \frac{c+\nicefrac{2}{\alpha}}{1-\alpha}\right]+c+1\right)\cdot\textsc{opt-w}<51\cdot\textsc{opt-w}.
\]
\medskip

We now turn our attention to (d), and our findings will later be used to prove (b).
Let $S$ denote the unit square in which we want to pack a given set $\mathcal P$ of convex polygons, each of diameter at most $\delta$.
We choose $\alpha\mydef 1/2$ and choose the width of the mini-containers to be 1, i.e., $(c+1)w_{\max}=1$.\footnote{In fact, we can choose $\alpha$ depending on $\delta$ to get a denser packing.
It turns out that $\alpha\mydef 1-\frac{\sqrt{3\delta^3 - 4\delta^2 + \delta}}{1-\delta}$ is the optimal value, but we stick to $\alpha\mydef 1/2$ to keep the analysis simpler.}
We stack all mini-containers on top of each other.
Suppose 
that the total height of the mini-containers exceeds 1.
We prove that the total area of pieces is then more than the constant $\rho\mydef (1-5\delta)(1-2\delta)/4$.
For $\delta\mydef 1/10$, we get $\rho\mydef 1/10$, and the claim in the theorem follows.

We call a mini-container \emph{full} if the bounding box of all contained pieces has width more than $1-\delta$; otherwise there is room for a further piece.
A mini-container that is not full is called \emph{near-empty}.
We can pack the pieces such that in each height class, there is at most one near-empty mini-container.
Therefore, the total height of near-empty mini-containers is at most $\sum_i h_i=h_{\max}\sum_i \alpha^i \leq \delta /(1-\alpha)=2\delta$.
Since the mini-containers (full and near-empty) have a total height of more than $1$, the full mini-containers in $S$ have a total height of more than $1-2\delta$, and a total area of more than $1-2\delta$.

Let $B$ be the bounding box of a full mini-container of height $h_i$ (and area $h_i$) and denote the set of contained polygons by $\mathcal P'$. By~\cite{alt_convexOffline_JoCG,alt_convexOffline_corr}, we obtain
\[(1-\delta)h_i
\leq
\area(B)
\leq
 2/\alpha\cdot \left(\Area{\mathcal P'}+h_i w_{\max}\big)
 \leq 
 4\big(\Area{\mathcal P'}+h_i \delta\right).\]
Consequently, 
$
 \Area{\mathcal P'}
 \geq
 (1-5\delta)h_i/4
$ and thus the density in any full mini-container is at least $(1-5\delta)/4$.
As the total area of the full mini-containers in $S$ is more than $1-2\delta$, the total area packed into $S$ is more than $\rho\mydef (1-5\delta)(1-2\delta)/4$.

We use the described algorithm for square packing to prove statement (b).
We can guarantee a density of at least $\rho$ in all but one bin.
It follows that the number of bins is at most $\area(\mathcal P)/\rho+1$.
Clearly, the number of bins in the optimal solution is at least $\lceil \area(\mathcal P)\rceil$.
Therefore, the approximation ratio is at most
\[
\frac{\area(\mathcal P)/\rho+1}{\lceil \area(\mathcal P)\rceil}\leq 1/\rho+1/\lceil \area(\mathcal P)\rceil\leq 1/\rho + 1.
\]
Using $\delta\mydef 1/10$ yields the ratio $11$ as stated in the theorem.

We finally prove (c).
In order to obtain a bounding box with small perimeter, 
we consider the mini-containers in a greedy fashion (from largest to smallest height) and pack them on top of each other into stacks of height at most $H\mydef \sqrt{A_C}+h_{\max} $ and width $(c+1)w_{\max}$. Clearly the height of each stack except possibly the last one is at least $\sqrt{A_C}$.
Consequently,  the number of stacks is at most $\frac{\sqrt{A_C}}{(c+1)w_{\max}}+1$.
Hence, we obtain a bounding box with perimeter of at most $4\sqrt{A_C}+2 h_{\max}+2(c+1) w_{\max}$.
Note that the optimal perimeter \textsc{opt-p} is at least $\max\{2 w_{\max}+2h_{\max}, 4\sqrt{\area(\mathcal P)}\}$.
Using the AM--GM inequality, we get $h_{\max}w_{\max}\leq \textsc{opt-p}^2/16$, and we also have $\area(\mathcal P)\leq \textsc{opt-p}^2/16$.
We then get from \cref{eq:offlineeq} that
 \[
 A_C\leq
(1+\nicefrac{1}{c})\cdot \left[\frac{2}{\alpha}\cdot \frac 1{16}+ \frac{c+\nicefrac{2}{\alpha}}{1-\alpha} \cdot \frac 1{16}\right]\textsc{opt-p}^2.
 \]

 Finally, choosing $\alpha\mydef 0.5$ and $c\mydef 0.975$, we obtain
 \[
4\sqrt{A_C}+2 h_{\max}+2(c+1) w_{\max}\leq
\left(
4\sqrt{(1+\nicefrac{1}{c})\cdot \left[\frac{2}{\alpha}\cdot \frac 1{16}+ \frac{c+\nicefrac{2}{\alpha}}{1-\alpha} \cdot \frac 1{16}\right]}+1+c
\right)\cdot\textsc{opt-p}
 < 7.3\cdot \textsc{opt-p}.
 \]
 This completes the proof.
\end{proof}

\ifsa
\section*{Acknowledgements}

We thank Shyam Narayanan for improving the upper bound algorithm in \cref{thm:SortingUpperBoundBaby} to the asymptotically tight upper bound of $O(\sqrt{n})$.

We highly appreciate the BARC espresso machine, the comfy couches in the Creative Room, and the cucumbers.
We also thank Linda's travel grant, Denmark for opening the borders in the right moment after a COVID-19 lockdown, and the lenient train officer who trusted the great authority of Mikkel's letter of invitation, allowing for travelling to Copenhagen and making this work possible. We finally thank Rosie Cohen for introducing us to the beautiful American musical with strong connections to packing algorithms.

\fi

\newpage
\printbibliography
\end{document}